%% file: OnBorisovThm-arXiv.tex
\documentclass[]{amsart}
\usepackage{amsmath,amssymb,amsthm,latexsym}
\usepackage[authoryear,round]{natbib}

\renewcommand\Large{\@setfontsize\Large{15}{18}}
\usepackage{relsize}
\newcommand{\LargeMath}[1]{\mathlarger{\mathlarger{#1}}}

\usepackage{url}
\usepackage{tikz}
\usepackage{xspace}
\usepackage{graphicx}
\usepackage{psfrag}

\DeclareFontFamily{U}{MnSymbolC}{}
\DeclareSymbolFont{MnSyC}{U}{MnSymbolC}{m}{n}
\DeclareMathSymbol{\diamonddot}{\mathbin}{MnSyC}{"7E}
\DeclareFontShape{U}{MnSymbolC}{m}{n}{
    <-6>  MnSymbolC5
   <6-7>  MnSymbolC6
   <7-8>  MnSymbolC7
   <8-9>  MnSymbolC8
   <9-10> MnSymbolC9
  <10-12> MnSymbolC10
  <12->   MnSymbolC12}{}

\newcommand{\semph}[1]{\emph{\textbf{#1}}} 
\newcommand{\lemph}[1]{\emph{#1}}

\newcommand{\MATH}[1]{\ensuremath{#1}\xspace}
\newcommand{\MATHSF}[1]{\MATH{\mathsf{#1}}}
\newcommand{\MATHFRAK}[1]{\MATH{\mathfrak{#1}}}
\newcommand{\MATHBF}[1]{\MATH{\mathbf{#1}}}
\newcommand{\MATHCAL}[1]{\MATH{\mathcal{#1}}}

\newcommand{\Ev}{\MATHSF{Ev}}
\newcommand{\IM}{\MATHSF{IM}}
\newcommand{\IOb}{\MATHSF{IOb}}
\newcommand{\Q}{\MATHSF{Q}}
\newcommand{\M}{\MATHFRAK{M}}

\newcommand{\ax}[1]{\MATHSF{#1}}
\newcommand{\taxis}{\MATHBF{t}}
\newcommand{\de}{\MATH{\stackrel{\text{\tiny def}}{=}}}
\newcommand{\defiff}{\MATH{\ \stackrel{\text{\tiny def}}{\Longleftrightarrow}\ }}

\newcommand{\Id}{\MATHSF{Id}}
\newcommand{\Pres}{\MATHSF{P}}
\newcommand{\Co}{\MATHSF{C}}

\newcommand{\wrt}{w.r.t.\ }
\newcommand{\eg}{e.g.,\ }
\newcommand{\ie}{i.e.\ }
\newcommand{\cf}{cf.\ }

\theoremstyle{definition} 
\newtheorem{thm}{Theorem}

\newtheorem{cor}[thm]{Corollary} 
\newtheorem{prop}[thm]{Proposition}
\newtheorem{lem}[thm]{Lemma}

\newtheorem{rem}[thm]{Remark}

\newtheorem{que}{Question}
\newcommand{\QED}{\ensuremath{\hfill\Box}}

\newcommand*{\vv}[1]{\MATH{\vec{\mkern0mu#1\mkern2mu}}}

\newcommand{\vvu}{\vv{u}}
\newcommand{\vvv}{\vv{v}}
\newcommand{\vve}{\vv{e}}
\newcommand{\vvp}{\vv{p}} 
\newcommand{\vvq}{\vv{q}}
 
\newcommand{\BAX}{\ax{Borisov's\, Axioms}}

\newcommand{\E}{\MATHCAL{E}}
\newcommand{\C}{\MATHCAL{C}}
\newcommand{\GH}{\MATHSF{H}} 
\newcommand{\G}{\MATHSF{G}}
\newcommand{\MQ}{\MATHFRAK{Q}}
\newcommand{\R}{\mathbb{R}}
\newcommand{\W}{\mathbb{W}} 
\newcommand{\Poi}{\MATHSF{Poi}}
\newcommand{\cPoi}{{}_{\text{\normalsize \cc}}\MATHSF{Poi}}
\newcommand{\Gal}{\MATHSF{Gal}}
\newcommand{\Eucl}{\MATHSF{Eucl}} 
\newcommand{\cEucl}{{}_{\text{\normalsize \cc}}\MATHSF{Eucl}}
\newcommand{\Triv}{\MATHSF{Triv}}
\newcommand{\Trivi}{\MATHSF{Triv}^{\uparrow}}
\newcommand{\Eucli}{\MATHSF{Eucl}^{\uparrow}}
\newcommand{\cEucli}{{}_{\text{\normalsize \cc}}\MATHSF{Eucl}^{\uparrow}}
\newcommand{\Poii}{\MATHSF{Poi}^{\uparrow}}
\newcommand{\cPoii}{{}_{\text{\normalsize \cc}}\MATHSF{Poi}^{\uparrow}}
\newcommand{\Gali}{\MATHSF{Gal}^{\uparrow}}
\newcommand{\cGali}{{}_{\text{\normalsize \cc}}\MATHSF{Gal}^{\uparrow}}

\newcommand{\Sym}{\MATHSF{Sym}}
\newcommand{\Trf}{\MATHSF{Trf}}
\newcommand{\Tran}{\MATHSF{Tran}}
\newcommand{\Lin}{\MATHSF{Lin}}
\newcommand{\wl}{\MATHSF{wl}}

\newcommand{\cc}{\ensuremath{\mathit{c}}\xspace}
\newcommand{\sqed}[1]{\left|#1\right|^2}

\newcommand{\sct}{\textsc{t}}
\newcommand{\scx}{\textsc{x}}
\newcommand{\scy}{\textsc{y}}
\newcommand{\scz}{\textsc{z}}

\newcommand{\edot}{\odot}
\newcommand{\mdot}{\diamonddot}

\title[Groups of Transformations Implied by the Principle of Relativity]
      {Groups of Worldview Transformations Implied by Einstein's 
       Special Principle of Relativity over Arbitrary Ordered Fields}

      \author{Judit Madar\'asz}
      \address{Judit Madar\'asz, HUN-REN Alfr\'ed R\'enyi Institute of Mathematics, Budapest, Hungary}
      \email{madarasz.judit@renyi.hu}

      \author{Mike Stannett}
      \address{Mike Stannett, School of Computer Science, The University of Sheffield, Sheffield, UK}
      \email{m.stannett@sheffield.ac.uk}

      \author{Gergely Sz\'ekely}
      \address{Gergely Sz\'ekely, HUN-REN Alfr\'ed R\'enyi Institute of Mathematics, Budapest, Hungary  \& University of Public Service, Budapest, Hungary.}
      \email{szekely.gergely@renyi.hu}

\begin{document}

\maketitle

\begin{abstract}
In 1978, Yu.\ F.\ Borisov presented an axiom system using a few basic
assumptions and four explicit axioms, the fourth being a formulation
of the relativity principle; and he demonstrated that this axiom
system had (up to choice of units) only two models: a relativistic one
in which worldview transformations are Poincar\'e transformations and
a classical one in which they are Galilean.  In this paper, we
reformulate Borisov's original four axioms within an intuitively
simple, but strictly formal, first-order logic framework, and convert
his basic background assumptions into explicit axioms. Instead of
assuming that the structure of physical quantities is the field of
real numbers, we assume only that they form an ordered field. This
allows us to investigate how Borisov's theorem depends on the
structure of quantities.

We demonstrate (as our main contribution) how to
construct Euclidean, Galilean, and Poincar\'e models of Borisov's
axiom system over every non-Archimedean field. We also demonstrate the
existence of an infinite descending chain of models and transformation
groups in each of these three cases, something that is not possible
over Archimedean fields.

As an application, we note that there is a model of Borisov's axioms that satisfies the relativity
principle, and in which the worldview transformations are Euclidean isometries.
Over the field of reals it is easy to eliminate this model using
natural axioms concerning time's arrow and the absence of instantaneous motion. 
In the case of non-Archimedean
fields, however, the Euclidean isometries appear intrinsically as
worldview transformations in models of Borisov's axioms and neither
the assumption of time's arrow, nor the rejection of instantaneous
motion, can eliminate them.
\end{abstract}




\section{Introduction}
${}$\\

In his famous 1905 paper, \cite{Einstein} based his special
theory of relativity on two explicit postulates: (1) the
\emph{principle of relativity}, according to which all inertial
observers (coordinate systems) are equivalent; and (2) the \emph{light
  postulate}, according to which there exists a ``stationary''
coordinate system in which all light signals travel with the same
speed. A few years later, beginning in 1910,
\cite{Ign1910c,Ignatowsky,Ign1911b} attempted to simplify
the theory by removing the need for the light postulate. Assuming the
principle of relativity, some ideas from electrodynamics, and various
hidden assumptions, he deduced that the associated (homogeneous)
coordinate system transformations must be Lorentz
transformations. 

\clearpage

\cite{Frank+Rothe} then considered
both Einstein's and Ignatowsky's arguments in more detail, identifying
four assumptions (rather than two) made by Einstein. They argued not
only that the light postulate is unnecessary, but that two more of
these assumptions can also be deduced by considering how
transformations compose with one another and act on points and
lines. Their system was less restrictive than Ignatowsky's, because
transformations were shown to split into three distinct classes:
Galilean transformations; Lorentz transformations; and a rather
unusual class they called ``Doppler'' transformations
\citep[eqn.~129]{Frank+Rothe}.

In 1978, Yu.\ F.\ Borisov presented the axiom system with which we
shall mostly be concerned in this paper. In his axiom system,
  Borisov presented a few basic background assumptions and
four explicit axioms, the fourth being a formulation of the relativity
principle. Then he demonstrated that his axiom system
had (up to choice of units) only two models: a relativistic one in
which worldview transformations are Poincar\'e transformations (the
more general, inhomogeneous, counterparts of Ignatowsky's Lorentz
transformations), and a classical one in which they are Galilean~\citep{Borisov1978} 
(\cf~\cite[\S 10, pp.~60-61]{Guts}). Gyula D\'avid
subsequently showed (using a different framework) the existence
  of a model that also satisfies the relativity principle and in which
  worldview transformations are Euclidean isometries. He also proved
(over the field of reals) a characterization theorem stating that the
principle of relativity with some auxiliary assumptions implies that the
worldview transformations between inertial observers are either
Euclidean isometries, or else Galilean or Poincar\'e
transformations~\citep{dgyNoPh}. However, he eliminated the models
corresponding to Euclidean isometries by adding a new assumption, that
motion from one spatial location to another cannot be
instantaneous. Similarly, Euclidean isometries do not appear in
Borisov's models because they are eliminated by his assumption
\citep[\textsf{Axiom\,II}]{Borisov1978} that there is an arrow of time.

All of these studies implicitly assume that coordinates and other
physically observable quantities can be represented as values from the
real number field ($\mathbb{R}$), even though this assumption is not
well-founded. We have no empirical reason to make this assumption,
because all practical measurements yield only rational approximations
-- even quantum electro-dynamics (QED), widely regarded as the most
precisely tested physical theory so far, only boasts accuracies to
around 12 decimal digits~\citep{QEDAccuracy}. Indeed, one of our goals
in writing this paper is to investigate what properties of numbers are
actually needed if we want to model certain theories. For example, two
of us have previously shown elsewhere that special relativity can also
be modeled over the field of rational numbers~\citep{MSzRac}. We will
demonstrate, in fact, that special relativity theories defined over
\emph{non-Archimedean} fields are fundamentally different to those
defined over Archimedean fields like $\mathbb{R}$, because models
featuring Euclidean isometries as worldview transformations cannot be
eliminated, neither by introducing Borisov's arrow of time assumption,
nor by allowing D\'avid's ban on instantaneous motion. Moreover, the
group of worldview transformations contains an infinite descending
chain of proper worldview transformation subgroups -- this is in sharp
contrast to the situation when the field is Archimedean, where no such
descending chains are possible.

These results argue that non-standard analysis has an important role to play in the study of physical theories, since they show that one cannot
rule out --  on purely empirical grounds -- the possibility that infinitesimals and other non-standard values are physically relevant; they provide strong confirmatory evidence that the common preference for using real numbers to represent scalars and coordinates is one of convention rather than necessity. While a few authors have considered the use of non-standard analysis in mathematical physics and the study of stochastic processes (\eg \cite{nonstandard-physics,Nelson87,Alb87,BV88,CC95,Perlis2016}), most work in the field has understandably focussed on applications within mathematics itself. Our results suggest that further investigation of non-standard applications to mathematical physics is warranted, not just in terms of its usefulness as a tool for simplifying proofs, but in its own right: if non-standard scalars are taken to be physically meaningful, how does that change our understanding of the world around us?

As this brief summary illustrates, there is long history of research
addressing the logical foundations of relativity theory, with each
generation of researchers discovering implicit (unstated) assumptions
embedded in the work of their predecessors. We believe it is important
to avoid hidden assumptions wherever possible, so as to place physical
theory on secure logical and strictly mathematical foundations. Our
preferred approach is to express assumptions as simple, strictly
formal, first-order, explicit axioms, which are then used to support
formal proof-driven investigation. This provides the clarity we need
for ensuring that hidden assumptions are captured explicitly (\cf
\cite[\S Why FOL?]{BigBook} and~\cite[\S 11]{SzPhD}), while at the
same time making it relatively simple to verify whether the axioms are
strong enough to do the job we require of
them~\citep{GBT15,StannettNemeti}.

\subsection*{Contributions.} In this paper:

\begin{enumerate}
\item  We reformulate Borisov's original four axioms within an
  intuitively simple, but strictly formal, first-order logic
  framework, and convert his basic background assumptions into
  explicit axioms.
\item  Instead of assuming that the structure of physical
  quantities is the field of real numbers, we assume only that they
  form an ordered field. This allows us to investigate how Borisov's
  theorem depends on the structure of quantities.
\item  We characterize the groups of worldview transformations
  corresponding to models of (our reformulation of) Borisov's axiom
  system over any ordered field, \cf Theorems~\ref{thm:wtrf} and
  \ref{thm:model}. Using this characterization, we show that Borisov's
  \ax{Axiom\,III} is not needed to prove his main theorem, \cf
  Theorem~\ref{thm:B0}.
\item  We demonstrate (as our main contribution) how to construct
  Euclidean, Galilean, and Poincar\'e models of Borisov's axiom system
  over every non-Archimedean field, \cf Theorem~\ref{thm:non-arch}. In
  this theorem, we also demonstrate the existence of an infinite
  descending chain of models and transformation groups in each of
  these three cases, something that is not possible over Archimedean
  fields like $\mathbb{R}$.
\item  We show, in the case of non-Archimedean fields, that the
  Euclidean isometries appear intrinsically as worldview
  transformations in models of Borisov's axioms. Neither Borisov's
  assumption of time's arrow, nor D\'avid's rejection of instantaneous
  motion, can eliminate them.
\end{enumerate}

\section{Basic notations and axioms}

\subsection{The underlying logical framework}

Throughout this paper we consider models of the form
\[
  \M = \langle \Ev,\IM,\IOb,\Q, +,\cdot,\leq,\Co,\Pres \rangle
\]
where the various components are intended to have the following
interpretations:
\begin{itemize} 
\item $\Ev$ is a nonempty set of \semph{events};
\item $\IM$ is a nonempty set of \semph{inertial motions};
\item $\IOb$ is a nonempty set of \semph{inertial observers}
  (reference systems);
\item $\MQ=(\Q,+,\cdot,\leq)$ is a structure of \semph{quantities};
\item $\Co\subseteq\IOb\times\Ev\times\Q^4$ is a relation used to
  express \semph{coordinatization}; and
\item $\Pres\subseteq\IM \times \Ev$ is a relation used to express
  \semph{participation}.
\end{itemize} 

\noindent In other words, we use the following many-sorted first-order
logic language:
\begin{itemize}
\item $\Ev$, $\IM$, $\IOb$,and $\Q$ are four sorts representing
  different kinds of basic entities (events, inertial motions,
  inertial observers, and quantities);
\item $+$, $\cdot$, and $\le$ are the usual operation symbols and
  ordering relation, defined on the sort $\Q$;
\item $\Co$ is a relation of sort $\IOb\times\Ev\times\Q^4$, where the
  expression $\Co(k,e,\vvp)$ represents the idea that \lemph{inertial
    observer $k\in\IOb$ coordinatizes (``sees'') event $e\in\Ev$ at
    coordinate point $\vvp\in\Q^4$}; and
\item $\Pres$ is a relation of sort $\IM\times\Ev$, where the
  expression $\Pres(i,e)$ represents the idea that \lemph{inertial
    motion $i\in\IM$ participates in event $e\in\Ev$}, or in other
  words worldline of inertial motion $i$ contains event $e$.
\end{itemize}

We have chosen this formal language in accordance with Borisov's
choice of basic concepts, but using some notations of the
Andr\'eka--N\'emeti school to make it easier to connect the results of
this paper to the school's general project of logic based axiomatic
foundations of relativity theories. Axiomatic approaches to relativity
theory have extensive literature, see \eg \cite{AMNsamples}. Most of
these approaches use entirely different basic concepts, and hence it
is not at all straightforward to check whether these different axioms
systems capture the same physical theory. By showing the connection
between two radically different approaches, \cite{Comparing} takes an
important first step in bringing these sporadic axiom systems
together. See \cite{Friend15} and \cite{FriMol15} for general
discussions of Andr\'eka--N\'emeti school's project and methodology
from points of view of epistemological significance and role in
scientific explanation.

\subsection{Borisov's basic assumptions}
In his paper, Borisov declares four axioms
(\ax{Axiom\,I}--\ax{Axiom\,IV}), which are in turn based upon a number
of basic assumptions, which we here refer to as
$\ax{BA\,1}$--$\ax{BA\,4}$.

\vspace{12pt}
\begin{description}
\item[\boxed{\ax{BA\,1}}] \emph{$\MQ=(\Q,+,\cdot,\leq)$ is an ordered
  field in the sense of abstract algebra.}\footnote{In \citep{Borisov1978}, $\MQ$ is assumed to be the field of
      real numbers.}
\end{description}

\vspace{12pt}
\begin{description}
\item[\boxed{\ax{BA\,2}}] \emph{For every inertial observer
  $k\in\IOb$, the coordinatization\footnote{In \citep{Borisov1978}, 
	inertial observers and their coordinatizations
    are identified.} of $k$ is a bijection $\Co_k \colon \Ev\to \Q^4$.}
\end{description}
\vspace{12pt}

Given inertial observer $k\in\IOb$, the binary relation
$\Co_k\subseteq \Ev\times\Q^4$ between events and coordinate points
described in \ax{BA\,2} can be defined by
\[
   \Co_k(e,\vvp)\defiff\Co(k,e,\vvp).
\]

By \ax{BA\,2}, we can introduce the \semph{worldview transformations},
$f_{kh}$, \lemph{between inertial observers $h$ and $k$} as the
composition of bijections $\Co_k$ and $\Co^{-1}_h$, \ie
\[
  f_{kh}\de\Co_k\circ\Co^{-1}_h \colon \Q^4\to\Q^4.
\]
Again by \ax{BA\,2}, these worldview transformations are bijections
from $\Q^4$ to $\Q^4$, for which
\begin{equation}\label{eq:wwtrf}
  f_{mh}= f_{mk}\circ f_{kh}\qquad\text{ and }\qquad
  f_{kh}^{-1}=f_{hk}
\end{equation}
for all inertial observers $h,k,m\in\IOb$.

Let us define the \lemph{worldline\footnote{In \citep{Borisov1978}, inertial motions and their worldlines are
    identified.} of inertial motion $i\in\IM$} as
\[
  \wl(i)\de\big\{e\in\Ev \colon \Pres(i,e)\big\},
\]
and the \semph{worldline} \lemph{of inertial motion $i\in\IM$}
\semph{according to inertial observer $k\in\IOb$} as
\[
  \wl_k(i)\de\Co_k[\wl(i)].
\]
By \ax{BA\,2}, worldview transformations map worldlines to worldlines,
\ie for all inertial observers $h,k\in\IOb$ and every inertial motion
$i\in\IM$,
\begin{equation}\label{eq:wl}
  f_{kh}[\wl_h(i)]=\wl_k(i).
\end{equation}

\begin{figure}[!htb]
  \begin{center}
    \input{framework.tikz}
    \caption{Illustrating the concepts of coordinatization, worldlines
      and worldview transformations, and the relationship between
      them.}
  \end{center}
\end{figure}
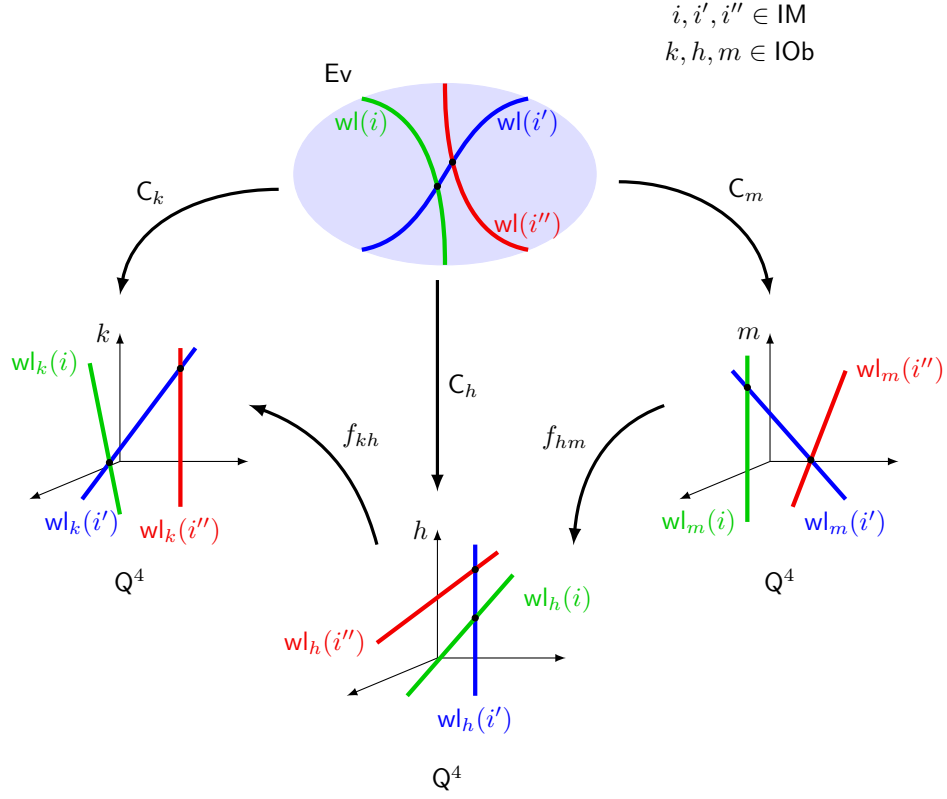

The \semph{time component} and \semph{space component} of
$\vvp=(p_0,p_1,p_2,p_3)\in\Q^4$ are defined respectively as
\[
   p_t\de p_0\qquad\text{and}\qquad \vvp_s\de (p_1,p_2,p_3).
\]
A set $\ell$ is called a \semph{line} if and only if there are
$\vvp,\vvv\in\Q^4$ with $\vvv \neq (0,0,0,0)$ such that $\ell=\{\vvp
+\lambda\cdot\vvv \colon \lambda\in\Q\}$.

The line $\ell$ is \semph{vertical} if and only if $\vvp_s=\vvq_s$ for every
$\vvp,\vvq\in\ell$.

\vspace{12pt}
\begin{description} 
\item[\boxed{\ax{BA\,3}}] \emph{For every\footnote{In our first-order
    logic language, quantifying over lines can be done by quantifying
    over pairs of coordinate points.} vertical line $\ell$ and
  inertial observer $k\in\IOb$, there is an inertial motion $i\in\IM$
  such that the worldline of $i$ according to $k$ is $\ell$, \ie
  $\wl_k(i)=\ell$.}
\end{description}
\vspace{12pt}

The line $\ell$ is of \semph{finite slope} if and only if there are
$\vvp,\vvq\in\ell$ such that $p_t\neq q_t$.

\vspace{12pt}
\begin{description} 
\item[\boxed{\ax{BA\,4}}] \emph{For every inertial motion $i\in\IM$
  and every inertial observer $k\in\IOb$, worldline $\wl_k(i)$ is a
  line of finite slope.}
\end{description}
\vspace{12pt}

\subsection{Borisov's Axioms}
We say that \lemph{inertial motion $i\in\IM$ is \semph{stationary} \wrt inertial observer $k$} if and only if $\wl_k(i)$ is a vertical line.  Then \lemph{inertial observer $h\in\IOb$ is \semph{at rest} \wrt inertial observer $k\in\IOb$}, if and only if, whenever inertial motion $i\in\IM$ is stationary \wrt $h$, then $i$ is also stationary \wrt $k$. 

\vspace{12pt}
\begin{description}
\item[\boxed{\ax{Axiom\, I}}] \emph{There exist inertial observers
  $k,h\in\IOb$, where $h$ is not at rest \wrt $k$.}
\end{description}
\vspace{12pt}

Given $\vvu=(u_0,\ldots,u_{n-1})\in\Q^n$ and
$\vvv=(v_0,\ldots,v_{n-1})\in\Q^n$, the \semph{scalar product} of
$\vvu$ and $\vvv$ is defined in the usual way:
\[
  \vvu\cdot\vvv \de u_0v_0 +\ldots + u_{n-1}v_{n-1}.
\]

The Euclidean length of certain vectors in $\Q^4$ may not exist
because $\MQ$ can be any ordered field, including for example the
field of rational numbers in which square roots are not always
defined. To avoid this problem, we use instead the
\semph{squared-Euclidean length} of $\vvv$, defined as
\[
  \sqed{\vvv}\de \vvv\cdot\vvv = v_0^2+\ldots+v_{n-1}^2.
\]

We call a transformation $T \colon \Q^4\to \Q^4$ a \semph{trivial
  transformation} if and only if (a) it takes vertical lines to
vertical lines; and (b) it is the composition of a translation and a
linear transformation preserving the squared-Euclidean length (or
equivalently, preserving the scalar product).  We say that
transformation $T:\Q^4\to\Q^4$ is \semph{orthochronous} if and only if
$T(1,0,0,0)_t > T(0,0,0,0)_t$. The set of orthochronous trivial
transformations is denoted by $\Trivi$.

\vspace{12pt}
\begin{description} 
\item[\boxed{\ax{Axiom\, II}}] 
${}$
\begin{itemize}
\item[a)] \emph{Given any inertial observer $k\in\IOb$ and any\footnote{Trivial transformations are affine ones
  and hence they can be represented by a $4\times 4$ matrix and a 
$4$-dimensional translation vector.  Therefore, in our first-order logic
  language, quantifying over trivial transformations can be done by
  quantifying over the $20$ quantity parameters
  representing those transformations.}
  orthochronous trivial transformation $\varphi\in\Trivi$, there is an
  inertial observer $h\in\IOb$ for which the worldview transformation
  between $h$ and $k$ is $\varphi$, \ie
  $f_{kh}=\varphi$.}
\item[b)] \emph{For all inertial observers $k,h\in\IOb$, if $h$ is at
  rest \wrt $k$, then $f_{kh}\in\Trivi$.}
\end{itemize}
\end{description}

\vspace{12pt}
\begin{description}
\item[\boxed{\ax{Axiom\, III}}] \emph{Given any inertial motion
  $i\in\IM$, any inertial observer $k\in\IOb$, and any orthochronous
  trivial transformation $\varphi\in\Trivi$, there is an inertial
  motion $i'\in\IM$ such that $\wl_k(i') =\varphi\big[\wl_k(i)\big]$.}
\end{description}

\vspace{12pt}
\begin{description}
\item[\boxed{\ax{Axiom\, IV}}] \emph{Given any inertial observers
  $k,k',h\in\IOb$, there is an inertial observer $h'\in\IOb$ such that
  $f_{hh'}=f_{kk'}$.}
\end{description}
\vspace{12pt}
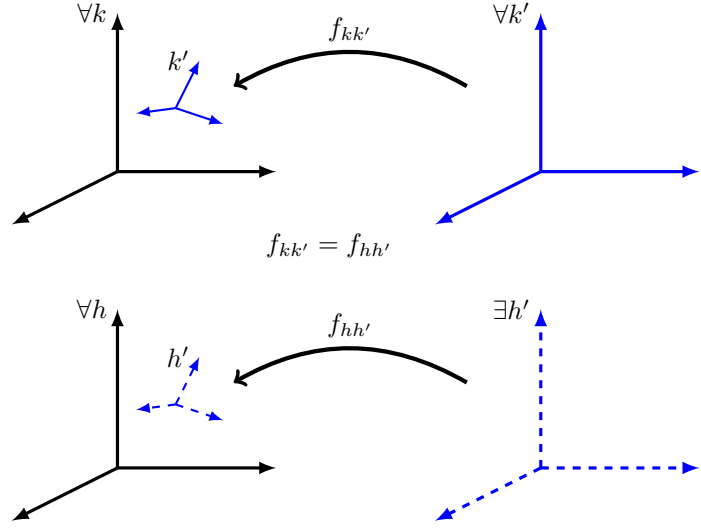
\begin{figure}[!tb]
  \begin{center}
    \input{ax4.tikz}
      \caption{Illustrating \ax{Axiom\,IV}.}
  \end{center}
\end{figure}

Overall, then, we capture Borisov's axiom system formally as:
\[
  \boxed{\BAX}\de \big\{ \ax{BA\,1}, \dots, \ax{BA\,4},
  \ax{Axiom\,I}, \dots,\ax{Axiom\,IV} \big\}.
\]
\vspace{12pt}

\subsection{\ax{Axiom\,IV} as a formulation of Einstein's Special Principle of Relativity}

The principle of relativity can be formalized in several different
ways, see \eg \cite{MSS3SPR,Gom15,GSz15}. In Borisov's axiom system, Einstein's
principle of relativity is captured by \ax{Axiom\,IV}. This is so
because in terms of worldview transformations \ax{Axiom\,IV} tells
that no inertial observer is distinguished by how its worldview can be
related to those of other observers. In this section, we are going to
explore some equivalent formulations of this central assumption. To do
so, for each inertial observer $k\in\IOb$, we define the
\semph{worldview of $k$} to be
\[
  \W_k\de \big\{f_{kh} \colon h\in\IOb\big\}.
\]
and the \lemph{set of \semph{worldview transformations}} by
\[
  \W \de \bigcup_{k \in \IOb}{\W_k} = \big\{f_{kh} \colon k,h\in\IOb\big\}.
\]
Where we wish to emphasize the underlying model $\M$ on which these
constructs are based, we will add the model name as a suffix, and
write, \eg $\W_{\M}$.
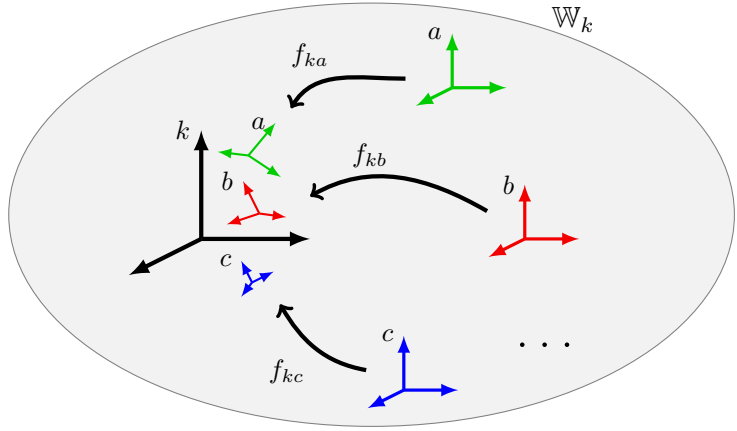
\begin{figure}[!hbt]
  \begin{center}
    \input{Wk.tikz}
          \caption{Illustrating the worldview $\W_k$ of inertial
          observer $k$.}
  \end{center}
\end{figure}

We first show (Proposition~\ref{prop:group}) that \ax{Axiom\,IV} is
equivalent to the claim that all observers have the same worldview,
which is in turn equivalent to saying that each worldview is a group
under composition.

\begin{prop}\label{prop:group} 
Assume \ax{BA\,2}. Then \ax{Axiom\,IV} is equivalent to each of the
following statements (and they are all equivalent to one another):

\begin{enumerate}
\renewcommand{\theenumi}{\roman{enumi}}
    \item $\W_k=\W_h$ for all $k,h\in\IOb$.
    \item $\W_k=\W$ for all $k\in\IOb$.
    \item $\W_k=\W$ for some $k\in\IOb$.
    \item $\W_k$ is closed under composition for all $k\in\IOb$. 
    \item $\W_k$ forms a group under composition for
      some $k\in\IOb$.
\end{enumerate}
\end{prop}

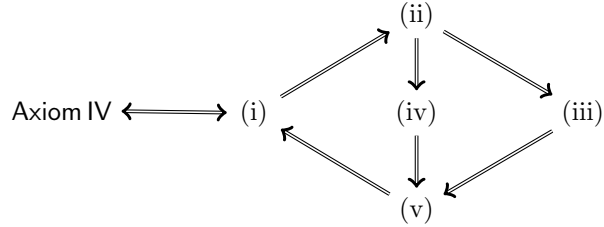
\begin{figure}[!htb]
  \begin{center}
    \input{proofdiagram.tikz}
  \end{center}
  \caption{Diagram illustrating the steps proving
    Proposition~\ref{prop:group}}
\end{figure}

\begin{proof}
\ax{Axiom\,IV} says that, given any $f_{kk'} \in \W_k$, there is an
$f_{hh'} \in \W_h$ satisfying $f_{kk'} = f_{hh'}$. So statement (i) is
simply a reformulation of \ax{Axiom\,IV} in terms of worldviews.

(i) $\implies$ (ii): follows because $\W = \bigcup_{h\in\IOb}\W_h$.

(ii) $\implies$ (iii): trivial as $\IOb$ is not empty.

Let us now prove the following:
\begin{equation}\label{eq:group}
  \W_k = \W \implies \W_k \text{ is a group}.
\end{equation}
We know that $\W_k=\W$ is closed under inverses because
$f^{-1}_{kh}=f_{hk}$ by \eqref{eq:wwtrf} (which follows from
\ax{BA\,2}). It remains to show that $\W$ is closed under
composition. To do so, choose any $f_{dc},f_{ba}\in\W$. Since
$\W_k=\W$, there are observers $a',d'\in\IOb$ such that
$f_{ka'}=f_{ba}$ and $f_{kd'}=f_{cd}$. Now \eqref{eq:wwtrf} yields
$f_{dc}\circ f_{ba}= f_{cd}^{-1}\circ f_{ba} = f_{kd'}^{-1}\circ
f_{ka'} = f_{d'k}\circ f_{ka'}=f_{d'a'}\in\W$, as required.

(iii) $\implies$ (v): follows trivially from \eqref{eq:group}.

(v) $\implies$ (i): We have $\W_h = f_{hk} \circ \W_k$ for all $k,
h \in \IOb$ since, by \eqref{eq:wwtrf}, $f_{hm} = f_{hk} \circ f_{km}$
for all $m \in \IOb$. Therefore, since $f_{hk}=f_{kh}^{-1}$ by
\eqref{eq:wwtrf}, we have $\W_h= f_{hk} \circ \W_k=f_{kh}^{-1}\circ
\W_k=\W_k$ for all $h\in\IOb$ if $\W_k$ is a group (as
$f_{kh}\in\W_k$).

(ii) $\implies$ (iv): follows trivially from \eqref{eq:group}.

(iv) $\implies$ (v): Since $\IOb$ is not empty, we only have to
prove that $\W_k$ is closed under composition and inverses for some
$k\in\IOb$. Since, by (iv), $\W_k$ is closed under composition it is
enough to show that it is also closed under inverses. To prove that,
let $f_{km}\in\W_k$ be arbitrary. We need to show that
$f_{mk}=f_{km}^{-1}\in\W_k$. Since $\W_m$ is also closed under
composition by (iv), we have $f_{mk}\circ f_{mk}\in\W_m$. Then there
is $h\in\IOb$ such that $f_{mk}\circ f_{mk}=f_{mh}$. Then
$f_{mk}=f_{mk}^{-1}\circ f_{mk}\circ f_{mk}=f_{km}\circ f_{mk}\circ
f_{mk}=f_{km}\circ f_{mh}=f_{kh}\in\W_k$.
\end{proof}

\section{Transformations}
${}$\\

Throughout this section, we assume \ax{BA\,1}, \ie that
$\MQ=(\Q,+,\cdot,\leq)$ is an ordered field. Let $\cc\in\Q$ and
  suppose $\cc>0$.

We define the \semph{$\cc$-scalar product} of vectors
$\vvp,\vvq\in\Q^4$ as:
\begin{equation*}
  \vvp\edot_{\cc}\vvq\de \cc p_t\cdot \cc q_t + \vvp_s\cdot\vvq_s,
\end{equation*}
and analogously the \semph{$\cc$-Minkowski scalar product} as:
\begin{equation*}
  \vvp\mdot_{\cc}\vvq\de \cc p_t\cdot \cc q_t - \vvp_s\cdot \vvq_s .
\end{equation*}
(Note that the $1$-scalar product coincides with the usual scalar product.) As usual, we write $|v|\de\max\{-v,v\}$ for the absolute value of $v \in \Q$.

We say that a function $T \colon \Q^4\to \Q^4$ is a \semph{linear
  $\cc$-Euclidean isometry} if and only if it is a linear transformation which
preserves the \cc-scalar product, \ie
\[
T\vvp\edot_{\cc} T\vvq=\vvp\edot_{\cc} \vvq 
\]
for every $\vvp,\vvq\in\Q^4$.   We call $T$ a \semph{$\cc$-Euclidean
  isometry} if and only if it is a composition of a linear $\cc$-Euclidean
isometry and a translation.  

We say that $T$ is a \semph{linear Galilean transformation} if and only if $T$ is a linear transformation and, for every
$\vvp,\vvq\in\Q^4$,
\[
|(T\vvp)_t|=|p_t|
\qquad\text{ and }\qquad p_t=q_t=0\implies T\vvp\cdot T\vvq =\vvp\cdot\vvq.
\]
$T$ is a \semph{Galilean transformation}  if and only if it is a composition of a linear
Galilean transformation and a translation.\footnote{In the literature it
  is customary to assume that Galilean transformations preserve time
  orientation, but here we would like to speak also about time
  reversing Galilean transformations.  Hence here it is more natural
  to introduce them this way.}  
  
We call $T$ a \semph{linear $\cc$-Poincar\'e transformation} if and only if it is a linear transformation
which preserves the \cc-Minkowski scalar product, \ie
\[
T\vvp\mdot_{\cc} T\vvq=\vvp\mdot_{\cc} \vvq, 
\]
for every $\vvp,\vvq\in\Q^4$.  $T$ is a \semph{$\cc$-Poincar\'e
  transformation} if and only if it is a composition of a linear
$\cc$-Poincar\'e transformation and a translation.

In the particular case when $\cc = 1$, we sometimes revert to the more familiar standard terminology, viz. $T$ is a \semph{Euclidean isometry} if and only if it is a 1-Euclidean isometry, 
and a \semph{Poincar\'e  transformation} if and only if it is a
1-Poincar\'e transformation. 

Finally, we note that $T$ is a trivial transformation if and only if it is a Euclidean isometry taking vertical lines to vertical lines.

Throughout this paper we investigate the various groups corresponding to these different transformation classes. The notation we use for each of these transformation groups is shown in Table \ref{tab:group-notations}.

\begin{center}
\begin{table}
  \caption{Transformation group notations used in this paper.}
\label{tab:group-notations}
    \begin{tabular}{rcl}
       \hline\hline
      $\Triv_{\MQ}$&$\quad$& Trivial transformations over $\MQ$\\
      $\Eucl_{\MQ}$&&  Euclidean isometries over $\MQ$\\
       $\Poi_{\MQ}$&&   Poincar\'e transformations over $\MQ$\\
      $\cEucl_{\MQ}$&& $\cc$-Euclidean isometries over $\MQ$\\
      $\cPoi_{\MQ}$&&   $\cc$-Poincar\'e transformations over $\MQ$ \\
      $\Gal_{\MQ}$&&  Galilean Transformations over $\MQ$ \\
       \hline\hline
    \end{tabular}
  \end{table}
\end{center}

The orthochronous variants of these sets are respectively denoted by
$\Trivi_{\MQ}$, $\Eucli_{\MQ}$, $\Poii_{\MQ}$, $\cEucli_{\MQ}$,
$\cPoii_{\MQ}$ and $\Gali_{\MQ}$.

As one would expect, all of these sets (except $\cEucli_{\MQ}$) are
groups under composition. 
\label{subset-not-subgroup}
To see that $\cEucli_{\MQ}$
is \emph{not} a group, notice that while it is closed under inverses, it is not closed under composition. For example, for the case
  $\cc=1$, let $R$ be a rotation of $45^\circ$ in the first two
  coordinates which leaves the other two coordinates fixed. Then $R\in\cEucli_{\MQ}$,   but $R\circ R\not\in\cEucli_{\MQ}$.
	
	Let us also note that
\begin{equation*}
  \Triv_{\MQ}=\bigcap_{\cc>0}\cPoi_{\MQ}\cap\bigcap_{\cc>0}\cEucl_{\MQ}\cap \Gal_{\MQ}
\end{equation*}
and hence
\begin{equation*}
  \Trivi_{\MQ}=\bigcap_{\cc>0}\cPoii_{\MQ}\cap\bigcap_{\cc>0}\cEucli_{\MQ} \cap\Gali_{\MQ}.
\end{equation*}

Let us introduce the following notations for the
\semph{squared $\cc$-Euclidean length} and \semph{squared $\cc$-Minkowski length}:
  \begin{equation*}
   \|\vvp\|^2_\cc\de \vvp\edot_\cc\vvp=\cc^2p_t^2+|\vvp_s|^2\quad \text{ and }\quad
   \|\vvp\|^2_{\cc,\mu}\de \vvp\mdot_\cc\vvp=\cc^2p_t^2-\sqed{\vvp_s}
\end{equation*}
We note that squared 1-Euclidean length coincides with the squared-Euclidean length.

Finally, we introduce the following notations for the standard basis of $\Q^4$:
\begin{equation*}
  \vve_{\sct}\de(1,0,0,0),\ 
  \vve_{\scx}\de(0,1,0,0),\ 
  \vve_{\scy}\de(0,0,1,0),\ 
  \vve_{\scz}\de(0,0,0,1).
\end{equation*}

\begin{prop} 
\label{prop:Eucl}
Let $L$ be a linear transformation. Then the following are equivalent:

\begin{enumerate}
\renewcommand{\theenumi}{\roman{enumi}}
\item $L\in\cEucl_{\MQ}$.
\item Given any $i,j\in\{\sct,\scx,\scy,\scz\}$,
\[
L\vve_i\edot_{\cc}L\vve_j=\vve_i\edot_{\cc} \vve_j=
    \begin{cases}
      \cc^2 &\text{ if } i=j =\sct \\
      1     & \text{ if } i=j\neq\sct\\  
     0 & \text{ if } i\neq j.
    \end{cases}
\]
\item Given any $\vvp\in\Q^4$$,
  \|L\vvp\|^2_\cc=\|\vvp\|^2_\cc$, \ie $L$ preserves the squared
  \cc-Euclidean length.
\end{enumerate}
\end{prop}

\begin{proof}
  (i) $\implies$ (ii) and (i) $\implies$ (iii) hold by
  definition (this is what we mean by \cc-Euclidean transformations,
  $\cc$-scalar product, and squared $\cc$-Euclidean length).

  (ii) $\implies$ (i) follows by direct calculation in the standard
  basis because $L$ is a linear transformation and
  $\edot_\cc \colon \Q^4\times\Q^4\rightarrow\Q$ is a bilinear function.

  (iii) $\implies$ (i) follows because $\edot_\cc$ is symmetric,
  whence we can write
  \[
  \vvp\edot_\cc\vvq =\frac{\|\vvp+\vvq\|^2_\cc
    -\|\vvp\|^2_\cc-\|\vvq\|^2_\cc}{2} \quad\text{ for every
  }\vvp,\vvq\in\Q^4
  \]
  and the claim follows immediately.
\end{proof}

The proofs of Proposition~\ref{prop:Gal} and \ref{prop:Poi} below are
analogous and we leave the details to the reader.

\begin{prop} 
\label{prop:Gal} Let $L$ be a linear transformation. 
Then the following are equivalent:

\begin{enumerate}
\renewcommand{\theenumi}{\roman{enumi}}
\item  $L\in\Gal_{\MQ}$.
\item $(L\vve_\sct)_t = \pm 1,\ (L\vve_\scx)_t= (L\vve_\scy)_t=
(L\vve_\scz)_t=0$;  and 
\[
 L\vve_i\cdot L\vve_j =  \vve_i\cdot \vve_j=
\begin{cases}
1 & \text{ if } i=j\\
0 & \text{ if } i \neq j 
\end{cases}\ \text{ for all }\ i,j\in\{\scx,\scy,\scz\}.
\]
\item Either $(L\vvp)_t=p_t$ for all $\vvp\in\Q^4$ or else
  $(L\vvp)_t=-p_t$ for all $\vvp\in\Q^4$; and $p_t=0\implies
  \sqed{L\vvp} =\sqed{\vvp}.$ \QED
\end{enumerate}
\end{prop}

\begin{prop}
\label{prop:Poi}
  Let $L$ be a linear transformation. Then the following are
  equivalent:
  
\begin{enumerate}
\renewcommand{\theenumi}{\roman{enumi}}
\item $L\in \cPoi_{\MQ}$.
\item Given any $i,j\in\{\sct,\scx,\scy,\scz\}$,
\[
L\vve_i\mdot_{\cc}L\vve_j= \vve_i\mdot_{\cc} \vve_j=
    \begin{cases}
      \cc^2 &\text{ if } i=j=\sct \\
      -1& \text{ if } i=j\neq \sct\\
      0 & \text{ if } i\neq j.
    \end{cases}
    \]
\item Given any $\vvp\in\Q^4$,
  $\|L\vvp\|^2_{\cc,\mu}=\|\vvp\|^2_{\cc,\mu}$, \ie $L$ preserves
  the squared \cc-Minkowskian length. \QED
\end{enumerate}
\end{prop}

\section{Theorems}
\label{sec:theorems}
${}$\\

In this section we outline the various theorems that we'll be proving in this paper; the proofs follow in \ref{sec:constructions} and \ref{proof:non-arch}. Our main result is Theorem \ref{thm:non-arch}, showing how Borisov's theorem changes when fields other than $\R$ are considered. Whenever we speak about a set $\G$ of transformations as a group, we mean the group $\langle \G,\circ\rangle$. Recall that if we assume \ax{BA\,2} and \ax{Axiom\,IV}, then $\W_k=\W$ for every $k\in\IOb$ and $\W$ is a group (by Proposition~\ref{prop:group}). Recall also what Borisov's theorem states:

\begin{center}\parbox{.9\textwidth}{%
\underline{\textsc{Borisov's Theorem}}~\citep{Borisov1978}\newline{}
Suppose that \BAX hold and that $\MQ$ is the ordered field $\R$ of reals.  Then there are just two possibilities: either
\begin{align*}
\W&=\Gali_\R, \text{ or else }\\
\W&=\cPoii_\R \text{ for some }\cc>0.
\end{align*}
}
\end{center}
\label{rem1} 
In the first-order logic framework developed in
  this paper, it can be shown that Borisov's Theorem remains true if
we omit the assumption that $\MQ$ is the ordered field $\R$ of reals,
and simply assume instead that $\MQ$ is an Archimedean ordered field
in which every positive number has a square root (we omit the
details).
 It is an open question whether the statement remains
true if we omit the assumption that positive numbers have square
roots. Note, however, that if
we drop the assumption that $\MQ = \R$ without adding anything new, then Borisov's Theorem is no longer valid, because $\W$ can then be a \lemph{proper}
subgroup of $\Gali$ and $\cPoii$.  Indeed, $\W$ can even be an
orthochronous subgroup of $\cEucl$.

As usual, we will write $\GH\le\G$ to mean that \GH is a subgroup of \G, and $\GH<\G$
to indicate that it is a proper subgroup. Analogously, we write
$\M<\M'$ to mean that $\M$ is a proper submodel of $\M'$.

\begin{thm}[\sc{Main Theorem}]
\label{thm:non-arch} Let $\MQ$ be any
non-Archimedean ordered field.  Then, for every $\cc>0$, there are
models $\M^E$, $\M^G$ and $\M^P$ of \BAX over \MQ such that
\begin{align*}
  \Trivi_{\MQ}&< \W_{\M^E} \subset\cEucli_{\MQ}, \\
  \Trivi_{\MQ}&< \W_{\M^G} < \Gali_{\MQ},   \\
  \Trivi_{\MQ}&< \W_{\M^P} < \cPoii_{\MQ}.
\end{align*}
Moreover, in each of these three cases, there is a strictly
  descending countably infinite chain of models
\begin{align*}
& \ldots <  \M^E_i  < \ldots < \M^E_1<\M^E_0,\\
& \ldots <  \M^G_i  < \ldots < \M^G_1<\M^G_0,\\
& \ldots <  \M^P_i  < \ldots < \M^P_1<\M^P_0 
\end{align*}
over $\MQ$  of \BAX such that
\begin{align*}
  & \Trivi_{\MQ}<\ldots<\W_{\M^E_i}<\ldots < \W_{\M^E_1} <\W_{\M^E_0}\subset\cEucli_{\MQ}, \\
  & \Trivi_{\MQ}<\ldots< \W_{\M^G_i}<\ldots < \W_{\M^G_1} <\W_{\M^G_0} <\Gali_{\MQ},  \\
  & \Trivi_{\MQ}< \ldots< \W_{\M^P_i}<\ldots < \W_{\M^P_1} <\W_{\M^P_0} <\cPoii_{\MQ}. 
\end{align*}
\QED
\end{thm}

Note the use of \emph{subset} ($\subset$) -- as opposed to \emph{subgroup} ($<$) -- inclusion in the cases involving $\cEucli_{\MQ}$  (and likewise below); this is because, as observed on page~\pageref{subset-not-subgroup}, $\cEucli_{\MQ}$ is not a group.

\begin{thm}
\label{thm:groupmodel}
Let $\MQ$ be an ordered field and let $\G$ be a group for which
either
\begin{align*}
  \Trivi_{\MQ} &< \G \subset \cEucli_{\MQ}, \text{ or}\\
  \Trivi_{\MQ} &< \G \le       \Gali_{\MQ}, \text{ or}\\
  \Trivi_{\MQ} &< \G \le       \cPoii_{\MQ}.
\end{align*}
  Then there is a model \M of \BAX over $\MQ$ such that $\W_{\M}=\G$.
\QED
\end{thm}

\medskip
We also show that Borisov's Theorem 
remains true if we remove \ax{Axiom\,III} from \BAX: 
\begin{thm}\label{thm:B0}
Assume $\BAX\setminus\{\ax{Axiom\,III}\}$ and that $\MQ$ is the
ordered field $\R$ of reals.  Then there are two possibilities: either
\begin{align*}
  \W&=\Gali_{\R}, \text{ or else} \\
  \W&=\cPoii_{\R}\ \text{ for some }\cc>0.
\end{align*}
\QED
\end{thm}

\begin{rem}\label{rem2}
  If we don't require $\MQ$ to be the field $\R$ of reals, 
then using results in \cite{MSSnop}
the
  following can be proven:
 \label{thm-otherpaper}
  Assume $\BAX\setminus\{\ax{Axiom\,III}\}$ and that every positive
  number has a square root in $\MQ$.  Then either
  \begin{align*}
    \Trivi_{\MQ} &< \W\subset\cEucli_{\MQ}, \text{ or}\\
    \Trivi_{\MQ} &< \W\leq \Gali_{\MQ}, \text{ or}\\
    \Trivi_{\MQ} &< \W\leq\cPoii_{\MQ}\ \text{ for some}\ \ \cc>0.
  \end{align*}
  Moreover, the statement remains true even if we replace \ax{BA\,4}
  with the more general assumption: ``For every inertial motion
  $i\in\IM$ and every inertial observer $k\in\IOb$, $\wl_k(i)$ is a
  line''. It
  is an open question if they remain true if we omit the assumption
  that positive numbers have square roots.
\end{rem}

The next result, Corollary~\ref{cor-1}, is a direct consequence of
Borisov's Theorem and \ref{thm:groupmodel}.

\begin{cor}\label{cor-1} 
  Let $\MQ$ be the ordered field $\R$ of reals. Then there is
  \emph{no} group $\G$ for which
  \begin{align*}
    \Trivi_{\R} &< \G \subset \cEucli_{\R}, \text{ or}\\
    \Trivi_{\R} &< \G <\Gali_{\R}, \text{ or}\\
    \Trivi_{\R} &< \G <\cPoii_{\R}.
  \end{align*}
  \QED
\end{cor}

Corollary~\ref{cor-1} fails for more general choices of \MQ because
(by Theorem~\ref{thm:non-arch}) for all non-Archimedean fields $\MQ$,
there are strictly descending countably infinite chains of subgroups
such that
\begin{align*}
  \Trivi_{\MQ}&<\ldots < \G^E_i <\ldots<\G^E_1<\G^E_0 \subset \cEucli_{\MQ},\\ 
  \Trivi_{\MQ}&<\ldots < \G^G_i <\ldots<\G_1^G<\G_0^G <\Gali_{\MQ}, \text{ and}\\
  \Trivi_{\MQ}&<\ldots < \G^P_i <\ldots<\G^P_1<\G^P_0<\cPoii_{\MQ}.
\end{align*}

\section{Worldview transformation groups and model constructions}
\label{sec:constructions}
${}$\\

\noindent
The symmetric group over $\Q^4$ is denoted by $\Sym(\Q^4)$. The \semph{time-axis} is defined as 
\begin{equation*}
 \taxis  \de \big\{ (t,0,0,0)\in \Q^4 \colon  t \in  \Q\big\}.  
\end{equation*}

\begin{thm}\label{thm:wtrf}
Assume $\BAX\setminus\{\ax{Axiom\, III}\}$. Then $\W$ is a group satisfying:
\begin{enumerate}
\renewcommand{\theenumi}{\roman{enumi}}
  \item\label{wtrf1} $\Trivi_{\MQ}<\W\le \Sym(\Q^4)$;
  \item\label{wtrf2} for every $f\in \W$, $f[\taxis]$ is a line of finite slope;
  \item\label{wtrf3} if $f\in \W$ takes vertical lines to vertical lines, then $f\in\Trivi_{\MQ}$.
\end{enumerate}
\end{thm}

\begin{proof}
\ref{wtrf1}  
We know from Proposition~\ref{prop:group} that $\W\le \Sym(\Q^4)$, and
\ax{Axiom\,II(a)} tells us that $\Trivi_{\MQ}\le \W$.  So it is enough to show that there is
some $f\in\W\setminus\Triv_{\MQ}$. By \ax{Axiom\,I}, there are $k,h\in\IOb$
such that $h$ is not at rest \wrt $k$, \ie there is $i\in\IM$ such
that $\wl_h(i)$ is a vertical line but $\wl_k(i)$ is not a vertical
line.  Since, by equation \eqref{eq:wl}, $f_{kh}$ takes $\wl_h(i)$ to
$\wl_k(i)$ and trivial transformations take vertical lines to vertical
ones, we have $f_{kh}\not\in\Triv_{\MQ}$, as required.

\ref{wtrf2}
Let $f_{kh}\in\W$. By \ax{BA\,3}, there is
$i\in\IM$ such that $\wl_h(i)=\taxis$. By \ax{BA\,4}, $\wl_k(i)$ is a
line of finite slope.  Now, by equation \eqref{eq:wl},
$f_{kh}[\taxis]=f_{kh}\left[\wl_h(i)\right]=\wl_k(i)$, as claimed.

\ref{wtrf3}
Finally, suppose $f_{kh}\in\W$ takes
vertical lines to vertical lines.  We will prove that $h$ is at rest
\wrt $k$.  Let $i$ be an arbitrary inertial motion which is stationary
according to $h$, \ie $\wl_h(i)$ is a vertical line. Since $f_{kh}$
maps $\wl_h(i)$ to $\wl_k(i)$, $\wl_k(i)$ is also a vertical
line. Thus $h$ is at rest \wrt $k$. Hence, by \ax{Axiom\,II(b)},
$f_{kh}\in\Trivi_{\MQ}$ as required.
\end{proof}

For every $\MQ=(\Q,+,\cdot,\le)$ and $\G\subseteq \Sym(\Q^4)$, 
we will construct a model $\M(\G)$ as follows:
\begin{equation*}  
  \Ev\de\Q^4,\quad 
  \IM\de\{g[\taxis] \colon g\in\G\},\quad 
  \IOb\de\G,
\end{equation*}
\begin{equation*}
\Co(k,e,\vvp)\defiff k(e)=\vvp,\qquad \Pres(i,e)\defiff e \in i, 
\end{equation*}
\begin{equation*}
  \mathfrak{M}(\G)\de\langle \Ev,\IM,\IOb,\Q,+,\cdot,\le,\Co,\Pres\rangle.
\end{equation*}

\begin{prop}
\label{prop:egyszeru}
Suppose $\G' \subset \G \subseteq \Sym(\Q^4)$.
Then $\M(\G')<\M(\G)$.
\end{prop}
\begin{proof} 
The proposition easily follows from the definitions of $\M(\G)$ and
$\M(\G')$.
\end{proof}

\begin{thm}\label{thm:model}
  Suppose $\MQ$ is an ordered field and $\G$ is a group for which:
\begin{enumerate}
\renewcommand{\theenumi}{\roman{enumi}}
  \item\label{awtrf1} $\Trivi_{\MQ}<\G\le \Sym(\Q^4)$;
  \item\label{awtrf2} For every $g\in \G$, $g[\taxis]$ is a line of
    finite slope;
  \item\label{awtrf3} If $g\in \G$ takes vertical lines to vertical
    lines, then $g\in\Trivi_{\MQ}$.
  \end{enumerate}
  Then $\M(\G)$ satisfies $\BAX$. Furthermore, $\W_{\M(\G)}=\G$.
\end{thm}

\begin{rem} Assume that $\MQ$ is an ordered field in which every positive
number has a square root and that $\G$ is a
group for which the conditions 
of Theorem~\ref{thm:model} hold. Then by Theorem~\ref{thm:model} and
and Remark~\ref{thm-otherpaper}, either
\begin{align*}
  \Trivi_{\MQ}&<\G\subset\cEucli_{\MQ},\text{ or}\\
  \Trivi_{\MQ}&<\G\leq \Gali_{\MQ},\text{ or}\\
  \Trivi_{\MQ}&<\G\leq\cPoii_{\MQ}\ \text{ for some } \cc>0.
\end{align*}
\QED
\end{rem}

\begin{proof}[Proof of Thm.~\ref{thm:model}]
It is easy to see that $\Co_k =k$ and $\wl(i)=i$ for every $k\in\IOb$
and $i\in\IM$, whence $f_{kh}=k\circ h^{-1}$ and $\wl_k(i)=k[i]$.
 
Axioms \ax{BA\,1} and \ax{BA\,2} hold by construction.
To prove \ax{BA\,3}, let $\ell$ be a vertical line and
$k\in\IOb=\G$. Then, there is $f\in\Trivi_{\MQ}$ such that
$f[\taxis]=\ell$.  It follows that $k^{-1}\circ f\in\G$ since
$\Trivi_{\MQ}\subseteq \G$ and $\G$ is a group. Let $i = (k^{-1}\circ
f)[\taxis]\in\IM$. Then $\wl_k(i) = k[i]=f[\taxis]=\ell$. Thus
\ax{BA\,3} holds.

To prove \ax{BA\, 4}, suppose $i\in\IM$ and $k\in\IOb=\G$. We have to
prove that $\wl_k(i)$ is a line of finite slope.  By the definition of
$\IM$ in $\M(\G)$, $i=g[\taxis]$ for some $g\in\G$.  Since $\G$ is a
group, $k\circ g\in\G$.  But now, by assumption \ref{awtrf2},
$\wl_k(i)=k[i]=(k\circ g)[\taxis]$ is a line of finite slope.  Thus
\ax{BA\,4} holds.

Let us now prove \ax{Axiom\,I}. By assumption \ref{awtrf1}, we can fix
some $g \in \G\setminus\Trivi_{\MQ}$.  By assumption \ref{awtrf3},
there is a vertical line $\ell$ such that $g[\ell]$ is not vertical.
Let us fix such an $\ell$. Let $\Id$ denote the identity
transformation of $\Q^4$. Then $\Id, g \in\IOb=\G$.  By the already
proven \ax{BA\,3}, there is $i\in\IM$ such that
$\wl_{\Id}(i)=\ell$. So, by equation \eqref{eq:wl},
$\wl_g(i)=f_{g\Id}\big[\wl_{\Id}(i)\big]=g[\ell]$.  Therefore, $i$ is
stationary \wrt inertial observer $\Id$ but is not stationary \wrt
inertial observer $g$, which proves \ax{Axiom\,I}.

To prove \ax{Axiom\,II(a)}, let $k\in\IOb$ and
$\varphi\in\Trivi_{\MQ}$. We have to prove that there is $h\in\IOb$
such that $f_{kh}=\varphi$.  Let $h = (k^{-1}\circ\varphi)^{-1}$.  Then
$h\in\G=\IOb$ since $\G$ is a group and $f_{kh}=k\circ
h^{-1}=\varphi$, which is what we wanted to prove.

To prove \ax{Axiom\,II(b)}, let $k,h\in\IOb$ be such that $h$ is at
rest \wrt $k$.  We have to prove that $f_{kh}\in\Trivi_{\MQ}$.  Let
$\ell$ be an arbitrary vertical line.  By the already proven
\ax{BA\,3}, there is $i\in\IM$ such that $\wl_h(i)=\ell$.  Since
$\ell$ is vertical, $i$ is stationary \wrt $h$.  Thus $i$ must be
stationary \wrt $k$ because $h$ is at rest \wrt $k$.  Hence $\wl_k(i)$
is also a vertical line. By equation \eqref{eq:wl}, $f_{kh}$ takes
$\wl_h(i)=\ell$ to vertical line $\wl_k(i)$.  Since $\ell$ was an
arbitrary vertical line, $f_{kh}$ takes vertical lines to vertical
ones.  Since $\G$ is a group, we have $f_{kh}=k\circ h^{-1}\in\G$.
Hence, by assumption \ref{awtrf3}, we have that
$f_{kh}\in\Trivi_{\MQ}$, which is what we wanted to prove.

To prove \ax{Axiom\,III}, let $i\in\IM$, $k\in\IOb$, and
$\varphi\in\Trivi_{\MQ}$.  We have to prove that there is an inertial
motion $i'\in\IM$ such that $\wl_k(i')=\varphi\big[\wl_k(i)\big]$. Let
$g\in\G=\IOb$ such that $i=g[\taxis]$. Then $k^{-1}\circ\varphi\circ
k\circ g\in \G$ because $\G$ is a group containing $\Trivi_{\MQ}$.
Let $i' = (k^{-1}\circ \varphi\circ k \circ g)[\taxis]$. Then $i'\in
\IM$ by the definition of $\IM$ in $\M(\G)$ and
\[
   k[i'] = k\big[(k^{-1}\circ\varphi\circ k\circ g)[\taxis]\big]
         = (\varphi\circ k) \big[g[\taxis]\big]=\varphi\big[k[i]\big].
\] 
Since $\wl_k(i')=k[i']$ and $\wl_k(i)=k[i]$, this proves that
$\wl_k(i')=\varphi\big[\wl_k(i)]$.

To prove \ax{Axiom\,IV}, let $k,k',h\in\IOb$. We have to prove that
there is $h'$ such that $f_{kk'}=f_{hh'}$, \ie $k\circ k'^{-1}=h\circ
h'^{-1}$. Such $h'$ exists because $\IOb = \G$ is a group: setting
$h'=k'\circ k^{-1}\circ h\in \IOb$ yields $f_{kk'}=f_{hh'}$.

It remains to show that $\W_{\M(\G)}=\G$. This is straightforward
because, by definition and because $\G$ is a group,
\[
    \W_{\M(\G)} = \big\{f_{kh} \colon k,h\in \G\}
		            = \{k\circ h^{-1} \colon k,h\in\G\big\}
								= \G
\]
as required.
\end{proof}

In the proof of Theorem~\ref{thm:groupmodel}, we will use the
following lemma.

\begin{lem} 
\label{lem:trivi}
Assume \ax{BA\,1}. Assume
$L\in\cEucli_{\MQ}\cup\Gali_{\MQ}\cup\cPoii_{\MQ}$ for some $\cc>0$
and that $L$ takes vertical lines to vertical lines. Then
$L\in\Trivi_{\MQ}$.
\end{lem}

\begin{proof}
Without loss of generality, we can assume that $L$ is linear. In this
case, we only have to prove that $L$ preserves the squared-Euclidean
length, \ie that $\sqed{L\vvp}=\sqed{\vvp}$ for every $\vvp\in\Q^4$.

Then $(L\vve_\sct)_t>0$ and $L\vve_\sct\in\taxis$ because $L$ is an
orthochronous linear transformation that takes vertical lines to
vertical lines. The only point on $\taxis$ with positive time
component and squared \cc-Minkowskian (squared \cc-Euclidean) length
of $\cc^2$ is $\vve_\sct$.  Therefore, $L\vve_\sct=\vve_\sct$ by
Propositions~\ref{prop:Eucl}, \ref{prop:Gal}, and \ref{prop:Poi}.

Let $\vvp=(p_t,p_x,p_y,p_z)\in\Q^4$ be arbitrary but fixed.  We have
to prove that $\sqed{L\vvp}=\sqed{\vvp}$. By the linearity of $L$
and $L\vve_\sct=\vve_\sct$, we have
\begin{equation}\label{eq:time}
  L(p_t,0,0,0)=(p_t,0,0,0).
\end{equation}
We also have
\begin{equation}\label{eq:space}
L(0,p_x,p_y,p_z)=(0,q_x,q_y,q_z)\quad\text{and}\quad
q_x^2+q_y^2+q_z^2=p_x^2+p_y^2+p_z^2
\end{equation}
for some $q_x,q_y,q_z\in\Q$ because of the following. If
$L\in\cPoii_{\MQ}$, then $L(0,p_x,p_y,p_z)_t=0$ because
$L(0,p_x,p_y,p_z)\mdot_\cc\vve_\sct = L(0,p_x,p_y,p_z)\mdot_\cc
L\vve_\sct =(0,p_x,p_y,p_z)\mdot_{\cc} \vve_\sct=0$ as $L$ is a linear
$\cc$-Poincar\'e transformation. Hence
$L(0,p_x,p_y,p_z)=(0,q_x,q_y,q_z)$ for some $q_x,q_y,q_z\in\Q$, and
$q_x^2+q_y^2+q_z^2=-\|L(0,p_x,p_y,p_z)\|^2_{\cc,\mu}=-\|(0,p_x,p_y,p_z)\|^2_{\cc,\mu}=p_x^2+p_y^2+p_z^2$
by Proposition~\ref{prop:Poi}.  A completely analogous proof based on
Proposition~\ref{prop:Eucl} shows that \eqref{eq:space} holds when
$L\in\cEucli_{\MQ}$. And if $L\in\Gali_{\MQ}$, then \eqref{eq:space}
holds by Proposition~\ref{prop:Gal}.

Thus, by linearity of $L$ and equations \eqref{eq:time} and
\eqref{eq:space}, we have
\begin{multline*}
    \sqed{L\vvp} = \sqed{L(p_t,0,0,0)+L(0,p_x,p_y,p_z)}
                   = \sqed{(p_t,q_x,q_y,q_z)}\\ 
                   = p_t^2+q_x^2+q_y^2+q_z^2 
                   = p_t^2+p_x^2+p_y^2+p_z^2
                   = \sqed{\vvp} 
  \end{multline*}
as claimed.
\end{proof}

\begin{proof}[Proof of Theorem~\ref{thm:groupmodel}]\label{proof:groupmodel}
Let $\G$ be a group for which $\Trivi_{\MQ}<\G$ and either
$\G\subset\cEucli_{\MQ}$, or $\G\le \Gali_{\MQ}$ or $\G\le
\cPoii_{\MQ}$.  Then $\G$ satisfies assumptions \ref{awtrf1} and
\ref{awtrf2} of Theorem~\ref{thm:model} because elements of $\G$ are
orthochronous linear transformations composed with translations. $\G$
also satisfies assumption \ref{awtrf3} of Theorem~\ref{thm:model} by
Lemma~\ref{lem:trivi}. By Theorem~\ref{thm:model}, $\M(\G)$ is a model
of \BAX and $\W_{\M(\G)}=\G$.
\end{proof}

\begin{proof}[Proof of Theorem~\ref{thm:B0}]\label{proof:B0}
Assume $\BAX\setminus\{\ax{Axiom\, III}\}$ and that $\MQ$ is the
ordered field $\R$ of reals.  Then, by Theorem~\ref{thm:wtrf}, $\W$ is
a group satisfying assumptions \ref{awtrf1}, \ref{awtrf2}, and
\ref{awtrf3} of Theorem~\ref{thm:model}. Hence $\M(\W)$ is a model of
\BAX and $\W_{\M(\W)}=\W$. The statement now follows by
Borisov's Theorem.
\end{proof}

\section{Proof  of the main theorem 
(Theorem~\ref{thm:non-arch})} 
\label{proof:non-arch}
${}$\\

Throughout this section, we assume \ax{BA\,1}, \ie that
$\MQ=(\Q,+,\cdot,\leq)$ is an ordered field.  Let us now prove our
main result, Theorem~\ref{thm:non-arch}. To do so, we first introduce some
additional notation. The \lemph{set of infinitesimals} is defined as follows:
\begin{equation*}
  \E\de\big\{x\in\Q \colon |nx|<1 \text{ for every  natural number } n \big\},
\end{equation*}
where $nx$ is an abbreviation for
$\underbrace{x+\ldots+x}_{n\text-times}$ when $x \in \Q$ and $n \in
\mathbb{N}$.

\noindent
A quantity $x\in\Q$ is said to be
\begin{itemize}
\renewcommand{\theenumi}{(\circ)}
\item \semph{infinitesimal} if and only if $x\in\E$;
\item \semph{unlimited} if and only if $1/x\in\E$; and 
\item \semph{limited} if and only if it is not unlimited. 
\end{itemize}
Let us note that $x$ is limited if and only if $|x|<n$ for some natural number
$n$.
  
When $\MQ$ is an Archimedean field, we have $\E=\{0\}$ (and hence there are no unlimited numbers), but in non-Archimedean fields there are infinitely many unlimited and infinitesimal numbers. 

We call a set $\C$ a \semph{cloud} (see Figure~\ref{fig:cloudball} below) if and only if the following hold:
\begin{itemize}
\item $\{0\} \subset \C \subseteq \E$;
\item if $x\in\C$ and $|y|\le|x|$, then $y\in\C$;
\item if $x\in\C$, then $2x\in\C$.
\end{itemize}
It is easy to see that (a) there is a cloud in $\MQ$ if and only if (b) $\E$ is a cloud if and only if (c) $\MQ$ is non-Archimedean.

\begin{lem}
\label{lem:smallest-cloud}
Suppose \MQ is non-Archimedean. Given any non-zero $\varepsilon \in
\E$, there is a smallest cloud $\C_\varepsilon$ containing
$\varepsilon$. It is the intersection of all clouds containing
$\varepsilon$, and moreover
\begin{equation*}\C_\varepsilon = \{
  \alpha\in\Q \colon |\alpha| \leq |n\varepsilon| \text{ for some } n \in
  \mathbb{N} \}.
\end{equation*}
\end{lem}
\begin{proof}
Let $\mathbb{C}_\varepsilon = \{ \C \colon \C \text{ is a cloud containing
} \varepsilon \}$ and define $\C_\varepsilon^* = \bigcap
\mathbb{C}_\varepsilon$. Note first that $\E$ is a cloud containing
$\varepsilon$, so $\mathbb{C}_\varepsilon$ is non-empty and
$\C_\varepsilon^* \subseteq \E$. 

We know that $\C \in
\mathbb{C}_\varepsilon \implies \{0,\varepsilon\} \subseteq \C$, so
$\{0,\varepsilon\} \subseteq \C_\varepsilon^*$. Thus, because $0 \neq \varepsilon$, we have $\{0\} \subset
\C_\varepsilon^* \subseteq \E$. 
If $x \in \C_\varepsilon^*$ then $x \in \C$ for all $\C \in
\mathbb{C}_\varepsilon$. Thus if $|y|\le|x|$, we have that $y \in \C$
for all $\C \in \mathbb{C}_\varepsilon$ and hence $y \in
\C_\varepsilon^*$. The proof that $2x \in \C_\varepsilon^*$ whenever
$x \in \C_\varepsilon^*$ is equally straightforward. This shows that $\C_\varepsilon^*$ is a cloud
containing $\varepsilon$, whence it is the smallest such cloud. 

Now define $\C_\varepsilon = \{ \alpha \colon |\alpha| \leq |n\varepsilon|
\text{ for some } n \in \mathbb{N} \}$. We need to show that
$\C_\varepsilon = \C_\varepsilon^*$. Notice first that
$\C_\varepsilon$ is certainly a cloud. First, it contains both 0 and
$\varepsilon$, and all elements of $\C_\varepsilon$ are
infinitesimal. Second, if $x \in \C_\varepsilon$, there exists some $n$ for
which $|x| \le |n\varepsilon|$, whence $|2x| \le |(2n)\varepsilon|$
and so $2x \in \C_\varepsilon$. And finally, if $|y|\le|x|$, then $|y|
\le |x| \le |n\varepsilon|$, so $y \in \C_\varepsilon$. From this it follows
that $\C_\varepsilon$ is a cloud containing $\varepsilon$, and hence
$\C_\varepsilon^* \subseteq \C_\varepsilon$.

It remains to show that $\C_\varepsilon \subseteq \C_\varepsilon^*$,
so choose any $\alpha \in \C_\varepsilon$ and any $\C \in
\mathbb{C}_\varepsilon$. By definition, there exists $n \in
\mathbb{N}$ such that $|\alpha| \le |n\varepsilon|$. Choose $m \in
\mathbb{N}$ such that $n \le 2^m$ and observe that $2^m\varepsilon \in
\C$, because $\varepsilon \in \C$ and \C is closed under
doubling. Since $|\alpha| \le |n\varepsilon| \le |2^m\varepsilon|$ it
follows that $\alpha \in \C$. Thus $\alpha$ belongs to every cloud
containing $\varepsilon$, and so $\alpha \in \C_\varepsilon^*$. It
follows that $\C_\varepsilon \subseteq \C_\varepsilon^*$, as required.
\end{proof}

\begin{cor}
\label{cor:infinitely-many-clouds}
Every non-Archimedean field contains infinitely many clouds.
\end{cor}
\begin{proof}
It is enough to show that for each $\varepsilon\in\E$, the smallest cloud containing $\varepsilon^2$ 
does not contain $\varepsilon$. We argue by contradiction. If $\varepsilon \in \C_{\varepsilon^2}$, then Lemma \ref{lem:smallest-cloud} tells us that there exists $n \in \mathbb{N}$ such that $|\varepsilon| \le |n\varepsilon^2|$. It follows that $1 \le |n\varepsilon|$, which contradicts the assumption that $\varepsilon\in\E$.
\end{proof}

\begin{lem}\label{lem:cloud}
  Let $\C$ be a cloud. Then
\begin{enumerate}
\renewcommand{\theenumi}{\roman{enumi}}
  \item $x,y\in\C\implies x+y\in\C,$
  \item $x\in\C\text{ and $y$ is limited}\implies xy\in\C.$
 \end{enumerate}
\end{lem}

\begin{proof}
  To prove (i), let $x,y\in\C$. Without loss of generality, we can assume that $|y|\le |x|$. Then, by the triangle inequality, $|x+y|\le |x| + |y| \le 2|x|$. Hence $x+y\in \C$.

To prove (ii), let $x\in\C$ and $y$ be limited. Without loss of generality, we can assume that $x>0$ and $y>0$. Since $y$ is limited there is a natural number $n$ such that $y<n$. Since $x>0$, we have $xy<nx$. There is a natural number $k$ such that $n<2^k$. Hence $0<xy<2^kx$. We have $2^kx\in \C$ since $\C$ is closed under doubling. Consequently, $xy\in\C$.     
\end{proof}

Let $\C$ be a cloud. Then the $\C$\semph{-ball} \lemph{around $\vve_{\sct}$} 
is defined as follows: 
\begin{equation*}
B_{\C}\de  \big\{\vvp\in\Q^4 \colon \sqed{\vvp-\vve_{\sct}}\le r^2\text{ for some } r\in\C
\big\}, 
\end{equation*}
see Figure~\ref{fig:cloudball} below.

 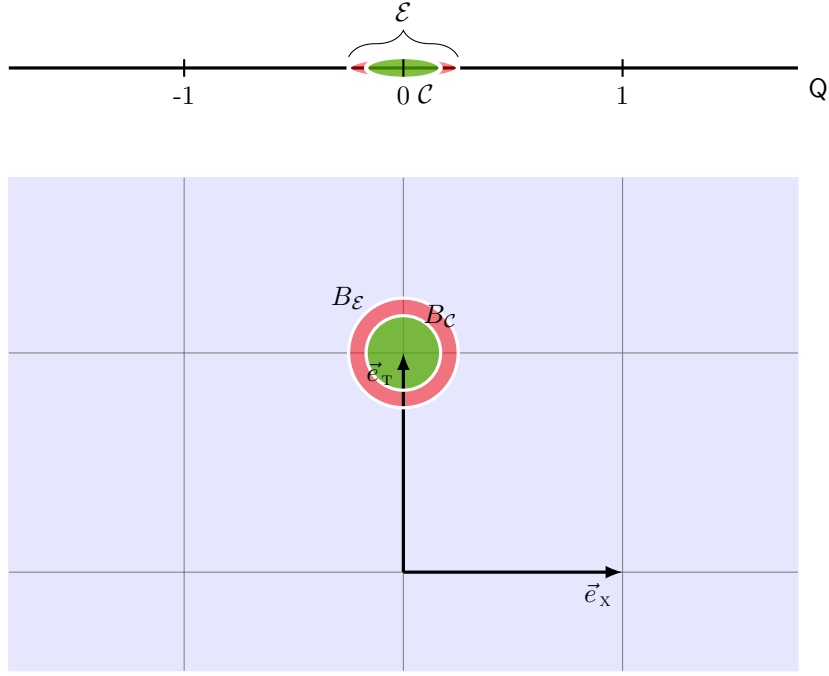
\begin{figure}
\begin{center}
   \input{cloudball.pgf}
\end{center}
\caption{\label{fig:cloudball}
Illustration for clouds 
$\C$ and $\E$, and the balls $B_{\C}$ and $B_{\E}$}
\end{figure}

By Proposition~\ref{prop:infball} below, the $\C$-balls and
``$\C$-boxes'' around $\vve_{\sct}$ are the same sets.
\begin{prop}\label{prop:infball}
  Let $\C$ be a cloud. Then
  \begin{equation*}
   (1+t,x,y,z)\in B_{\C} \iff t,x,y,z\in\C.
  \end{equation*}
\end{prop}

\begin{proof}
By definition, $(1+t,x,y,z)\in B_{\C}$ if and only if $t^2+x^2+y^2+z^2\le r^2$
for some $r\in \C$. So, if $(1+t,x,y,z)\in B_{\C}$, we have
$\max\{|t|,|x|,|y|,|z|\}\le r$ for some $r\in\C$. Consequently,
$t,x,y,z\in\C$. The converse follows because $t^2+x^2+y^2+z^2\le (2\max\{|t|,|x|,|y|,|z|\})^2$.
\end{proof}

\begin{prop}\label{prop:infballpres}
Let $f \colon \Q^4\to\Q^4$ be a linear transformation and let $\C$ be a
cloud. Assume that $\sqed{f(\vve_{\scx})}$, $\sqed{f(\vve_{\scy})}$,
$\sqed{f(\vve_{\scz})}$ are limited and $f(\vve_{\sct})\in B_\C$. Then
$f[B_\C]\subseteq B_\C$.
\end{prop}

\begin{proof}
Suppose $(1+t,x,y,z)\in B_\C$.  We need to show that $f(1+t,x,y,z)\in B_\C$.

Note first that $t,x,y,z\in\C$ by Proposition~\ref{prop:infball}. Since $f$ is linear,
\begin{equation*}
  f(1+t,x,y,z)=f(\vve_{\sct})+tf(\vve_{\sct})+xf(\vve_{\scx})+yf(\vve_{\scy})+zf(\vve_{\scz}).
\end{equation*}
For all $\vvp,\vvq\in\Q^4$, we have
\begin{equation}\label{sqtriangle}
  \sqed{\vvp+\vvq}\le
  2\sqed{\vvp}+2\sqed{\vvq}
\end{equation}
since $(a+b)^2\le 2a^2+2b^2$ for all $a,b\in\Q$. Consequently, we have 
\begin{align*}
  &\sqed{f(1+t,x,y,z)-f(\vve_{\sct})}
  \\ &=\sqed{tf(\vve_{\sct})+xf(\vve_{\scx})+yf(\vve_{\scy})+zf(\vve_{\scz})}
  \\ &\le 2\sqed{tf(\vve_{\sct})+xf(\vve_{\scx})}+
  2\sqed{yf(\vve_{\scy})+zf(\vve_{\scz})} \\&\le
  4\sqed{tf(\vve_{\sct})}+4\sqed{xf(\vve_{\scx})}+
  4\sqed{yf(\vve_{\scy})}+4\sqed{zf(\vve_{\scz})}\\ &\le
  4\max\left\{t^2\sqed{f(\vve_{\sct})},x^2\sqed{f(\vve_{\scx})},
  y^2\sqed{f(\vve_{\scy})},z^2\sqed{f(\vve_{\scz})}\right\}\\ &\le 
  r_0^2
\end{align*}
for some positive $r_0\in\C$ by Lemma~\ref{lem:cloud} since $f(\vve_{\sct})$ is also limited as $f(\vve_{\sct})\in B_\C$. By inequality
\eqref{sqtriangle},
\begin{equation*}
  \sqed{f(1+t,x,y,z)-\vve_{\sct}}\le 2\sqed{f(1+t,x,y,z)-f(\vve_{\sct})}+ 2\sqed{f(\vve_{\sct})-\vve_{\sct}}.
\end{equation*}
There is a positive $r_1\in\C$ such that
$\sqed{f(\vve_{\sct})-\vve_{\sct}}\le r_1^2$ because
$f(\vve_{\sct})\in B_\C$.  Therefore,
\begin{equation*}
  \sqed{f(1+t,x,y,z)-\vve_{\sct}}\le 2r_0^2+2r_1^2\le (2\max\{r_0,r_1\})^2.
\end{equation*}
Since $2\max\{r_0,r_1\}\in\C$, we have $f(1+t,x,y,z)\in B_\C$, which
is what we wanted to prove.
\end{proof}

Given any set \Trf of transformations, we will write $\Lin\Trf$ for the set of transformations in \Trf which are linear, \ie
\begin{equation*}
  \Lin\Trf\de\{ f\in \Trf \colon f \text{ is linear}\}.
\end{equation*}
In particular, if $\C$ is a cloud we will be interested in the following sets of linear transformations: 
\begin{align*}
  \Lin\Eucli_{\MQ}(\C) &\de \big\{f\in\Lin\Eucli_{\MQ} \colon
  f(\vve_{\sct})\in B_\C\big\},\\ \Lin\Gali_{\MQ}(\C) &\de
  \big\{f\in\Lin\Gali_{\MQ} \colon f(\vve_{\sct})\in
  B_\C\big\},\\ \Lin\Poii_{\MQ}(\C) &\de \big\{f\in\Lin\Poii_{\MQ} \colon
  f(\vve_{\sct})\in B_\C\big\}.
\end{align*}

\noindent We will also refer below to the following sets of affine transformations: 
\begin{align*}
 \Eucli_{\MQ}(\C) &\de \Tran_{\MQ}\circ\Lin\Eucli_{\MQ}(\C),\\
 \Gali_{\MQ}(\C)  &\de \Tran_{\MQ}\circ\Lin\Gali_{\MQ}(\C),\\
 \Poii_{\MQ}(\C)  &\de \Tran_{\MQ}\circ\Lin\Poii_{\MQ}(\C),
\end{align*}
where $\Tran_{\MQ}$ is the set of translations.

\begin{prop}\label{prop:ballpres}
  Let $\C$ be a cloud. Then
\begin{align*}
  \Lin\Eucli_{\MQ}(\C) &= \big\{f\in\Lin\Eucli_{\MQ} \colon f[B_\C]=
  B_\C\big\},\\ \Lin\Gali_{\MQ}(\C) &= \big\{f\in\Lin\Gali_{\MQ} \colon
  f[B_\C]= B_\C\big\},\\ \Lin\Poii_{\MQ}(\C) &= \big\{f\in\Lin\Poii_{\MQ} \colon
  f[B_\C]= B_\C\big\}.
\end{align*}
\end{prop}

\begin{proof}
Let $f \in\Lin\Eucli_{\MQ}\cup\Lin\Gali_{\MQ}\cup\Lin\Poii_{\MQ}$ and let
$\C$ be a cloud. It is enough to show that
\begin{equation*}
    f(\vve_{\sct})\in B_\C \iff f[B_\C]=B_\C.
  \end{equation*}
If $f\in\Lin\Poii_{\MQ}$, then 
\[
f(\vve_{\sct})_t =
f^{-1}(\vve_{\sct})_t\ \text{ and }\ \sqed{f(\vve_{\sct})_s} =
\sqed{f^{-1}(\vve_{\sct})_s}
\]
because
$f(\vve_{\sct})_t=\vve_{\sct}\mdot_{1}f(\vve_{\sct})=f^{-1}(\vve_{\sct})\mdot_1\vve_{\sct}=f^{-1}(\vve_{\sct})_t$,
$\sqed{\vvp_s}=p_t^2 -\|\vvp\|^2_{1,\mu}$ for every $\vvp\in\Q^4$ and
$\|f(\vve_{\sct})\|^2_{1,\mu}=\|\vve_\sct\|^2_{1,\mu}=\|f^{-1}(\vve_{\sct})\|^2_{1,\mu}$
by Proposition~\ref{prop:Poi}.
Thus
  \begin{multline*}
    \sqed{\vve_{\sct}-f(\vve_{\sct})}=(1-f(\vve_{\sct})_t)^2+|f(\vve_{\sct})_s|^2\\=(1-f^{-1}(\vve_{\sct})_t)^2+|f^{-1}(\vve_{\sct})_s|^2=\sqed{\vve_{\sct}-f^{-1}(\vve_{\sct})}.
  \end{multline*}
  If $f\in\Lin\Eucli_{\MQ}\cup\Lin\Gali_{\MQ}$,
  $$\sqed{\vve_{\sct}-f(\vve_{\sct})}=\sqed{f^{-1}(\vve_{\sct})-f^{-1}(f(\vve_{\sct}))}=\sqed{f^{-1}(\vve_{\sct})-\vve_{\sct}}$$
  because $f^{-1}$ is linear and preserves the Euclidean distance of
  $\vve_{\sct}$ and $f(\vve_{\sct})$ by Propositions~\ref{prop:Eucl}
  and \ref{prop:Gal}.

So
$\sqed{\vve_{\sct}-f(\vve_{\sct})}=\sqed{f^{-1}(\vve_{\sct})-\vve_{\sct}}$
in all three cases. Thus $f(\vve_{\sct})\in B_\C$ if and only if
$f^{-1}(\vve_{\sct})\in B_\C$ and hence $\Lin\Eucli_{\MQ}(\C)$,
$\Lin\Gali_{\MQ}(\C)$ and $\Lin\Poii_{\MQ}(\C)$ are closed under
taking inverse. Therefore, it is enough to show that
 $$f(\vve_{\sct})\in B_\C \iff f[B_\C]\subseteq B_\C.$$
Hence, by Proposition~\ref{prop:infballpres}, it is enough to show that
$f(\vve_{\sct})\in B_\C$ implies that $\sqed{f(\vve_{\scx})}$,  $\sqed{f(\vve_{\scy})}$, $\sqed{f(\vve_{\scz})}$ are limited.

If $f \in\Lin\Gali_{\MQ}$ or $f\in\Lin\Eucli_{\MQ}$,
$\sqed{f(\vve_{\scx})}=\sqed{f(\vve_{\scy})}=\sqed{f(\vve_{\scz})}=1$. Hence
they are limited.

To complete the proof in the last remaining case, let $f
\in\Lin\Poii_{\MQ}$.  We will prove that for every
$i\in\{\scx,\scy,\scz\}$, $\sqed{f(\vve_i)}\leq
\sqed{f(\vve_\sct)}$. To prove this, let $i\in\{\scx,\scy,\scz\}$ and let
\[
\vvp = f(\vve_\sct)\qquad \text{ and }\qquad \vvq = f(\vve_i).
\]
Then $p_t\cdot q_t - \vvp_s\cdot\vvq_s = f(\vve_\sct) \mdot_1
f(\vve_i) = \vve_\sct\mdot_1\vve_i =0$,
$p_t^2-\sqed{\vvp_s}=\|f(\vve_\sct)\|^2_{1,\mu}=\|\vve_\sct\|^2_{1,\mu}=1$
and
$q_t^2-\sqed{\vvq_s}=\|f(\vve_i)\|^2_{1,\mu}=\|\vve_i\|^2_{1,\mu}=-1$
by definition of Poincar\'e transformations and
Proposition~\ref{prop:Poi}.  Consequently,
\[
p_t\cdot q_t =\vvp_s\cdot\vvq_s, \quad \sqed{\vvp_s} = p_t^2-1,
\quad\text{and}\quad \sqed{\vvq_s} = q_t^2 + 1.
\]
Therefore, by the Cauchy-Schwarz inequality,\footnote{The Cauchy-Schwarz
  inequality over $\mathbb{R}$ has several elementary proofs, many of
  which remain valid over arbitrary ordered fields.}
\[
(p_t\cdot q_t)^2
     = (\vvp_s\cdot\vvq_s)^2
  \leq \sqed{\vvp_s}\cdot \sqed{\vvq_s}
     = (p_t^2-1) \cdot (q_t^2 +1)
     = (p_t\cdot q_t)^2-q_t^2+p_t^2-1.
\]
Thus $(p_t\cdot q_t)^2\leq (p_t\cdot q_t)^2-q_t^2+p_t^2-1$. This is
equivalent to $q_t^2 \leq p_t^2-1$, which is equivalent to
$q_t^2+(q_t^2 +1) \leq p_t^2 +(p_t^2-1)$, which is equivalent to
$q_t^2+\sqed{\vvq_s}\leq p_t^2 + \sqed{\vvp_s}$.  Thus $\sqed{\vvq}
\leq \sqed{\vvp}$, \ie $\sqed{f(\vve_i)} \leq \sqed{f(\vve_\sct)}$.
Now $\sqed{f(\vve_\sct)}$ is limited because $f(\vve_\sct)\in
B_\C$. Therefore, $\sqed{f(\vve_\scx)}$, $\sqed{f(\vve_\scy)}$ and
$\sqed{f(\vve_\scz)}$ are limited, too.
\end{proof}

\begin{lem}
\label{lem:therearetransformations}
Suppose $p_t,p_x,r_t,r_x\in\Q$ satisfy $p_t^2 + p_x^2 = r_t^2-r_x^2 =1$, 
and let $q$ be any value in $\Q$. Then there exist $f\in\Eucl_{\MQ}$,
$g\in\Gal_{\MQ}$ and $h\in\Poi_{\MQ}$ such that 
\begin{align*}
  f(\vve_{\sct}) &= (p_t,p_x,0,0) ,\\ 
  g(\vve_{\sct}) &= (1,q,0,0) ,\\
  h(\vve_{\sct}) &= (r_t,r_x,0,0) .
\end{align*}
\end{lem}

\begin{proof}
Let $f$ be the linear transformation that takes $\vve_{\sct}$,
$\vve_{\scx}$, $\vve_{\scy}$, $\vve_{\scz}$ to $(p_t,p_x,0,0)$,
$(-p_x,p_t,0,0)$, $\vve_{\scy}$, $\vve_{\scz}$, respectively.  Let $g$
be the linear transformation that takes $\vve_{\sct}$, $\vve_{\scx}$,
$\vve_{\scy}$, $\vve_{\scz}$ to $(1,q,0,0)$, $\vve_{\scx}$,
$\vve_{\scy}$, $\vve_{\scz}$, respectively.  Let $h$ be the linear
transformation that takes $\vve_{\sct}$, $\vve_{\scx}$, $\vve_{\scy}$,
$\vve_{\scz}$ to $(r_t,r_x,0,0)$, $(r_x,r_t,0,0)$, $\vve_{\scy}$,
$\vve_{\scz}$, respectively. Then, using Propositions~\ref{prop:Eucl},
\ref{prop:Gal}, and \ref{prop:Poi}, it is easy to check that
$f\in\Eucl_{\MQ}$, $g\in\Gal_{\MQ}$ and $h\in\Poi_{\MQ}$.
\end{proof}

\begin{thm}
\label{thm:szendvics1}
 Let $\C$ be a cloud. Then
\begin{align*}
  \Trivi_{\MQ} <& \Eucli_{\MQ}(\C) \subset \Eucli_{\MQ},\\
  \Trivi_{\MQ} <& \Gali_{\MQ}(\C)< \Gali_{\MQ},\\
  \Trivi_{\MQ} <& \Poii_{\MQ}(\C)< \Poii_{\MQ}.
\end{align*}
Moreover, there is a subcloud $\C' \subset \C$ such that 
\begin{align*}
  \Trivi_{\MQ} <& \Eucli_{\MQ}(\C') < \Eucli_{\MQ}(\C) \subset \Eucli_{\MQ},\\
  \Trivi_{\MQ} <& \Gali_{\MQ}(\C')  < \Gali_{\MQ}(\C)  < \Gali_{\MQ},\\
  \Trivi_{\MQ} <& \Poii_{\MQ}(\C')  < \Poii_{\MQ}(\C)  < \Poii_{\MQ}.
\end{align*}
\end{thm}

\begin{proof}
  To prove the first part of the theorem, it is enough to show that
\begin{align*}
   \Lin\Trivi_{\MQ} <& \Lin\Eucli_{\MQ}(\C) \subset \Lin\Eucli_{\MQ},
   \\ \Lin\Trivi_{\MQ} <& \Lin\Gali_{\MQ}(\C) < \Lin\Gali_{\MQ},
   \\ \Lin\Trivi_{\MQ} <& \Lin\Poii_{\MQ}(\C) < \Lin\Poii_{\MQ},
\end{align*}
because $\Trf=\Tran_{\MQ}\circ\Lin\Trf$ for each of the relevant
transformation groups $\Trf$.
  
Since $\Lin\Gali_{\MQ}$ and $\Lin\Poii_{\MQ}$ are groups, we have
\begin{equation*}
  \Lin\Gali_{\MQ}(\C)\le\Lin\Gali_{\MQ}\ \text{ and }\
  \Lin\Poii_{\MQ}(\C)\le\Lin\Poii_{\MQ}
\end{equation*}
by Proposition~\ref{prop:ballpres}.  By
Lemma~\ref{lem:therearetransformations}, there are
$g\in\Lin\Gali_{\MQ}$ and $h\in\Lin\Poii_{\MQ}$ such that
$g(\vve_\sct)=(1,1,0,0)$ and
$h(\vve_\sct)=\left(\frac{5}{3},\frac{4}{3},0,0\right)$.  Clearly for
such $g$ and $h$, we have $g\not\in\Lin\Gali_{\MQ}(\C)$ and
$h\not\in\Lin\Poii_{\MQ}(\C)$ because $(1,1,0,0),
\left(\frac{5}{3},\frac{4}{3},0,0\right)\not\in B_\C$.  Therefore,
$\Lin\Gali_{\MQ}(\C)<\Lin\Gali_{\MQ}$ and
$\Lin\Poii_{\MQ}(\C)<\Lin\Poii_{\MQ}$.

Since $\Lin\Eucli_{\MQ}$ does not form a group, we need to adopt a
different approach to prove that $\Lin\Eucli_{\MQ}(\C)$ is a group.
Let $f,g \in \Lin\Eucli_{\MQ}(\C)$.  Then $f\circ g\in\Lin\Eucl_{\MQ}$
and $f^{-1}\in\Lin\Eucl_{\MQ}$.  By Proposition~\ref{prop:ballpres},
we have $(f\circ g)[B_\C]=B_\C$ and $f^{-1}[B_\C]=B_\C$. Thus we have
$(f\circ g)(\vve_{\sct})_t>0$ and $f^{-1}(\vve_{\sct})_t>0$ because
$(f\circ g)(\vve_{\sct})\in B_\C$, $f^{-1}(\vve_{\sct})\in B_\C$ and
$\C\subseteq\E$.  Hence $f\circ g\in\Lin\Eucli_{\MQ}(\C)$ and
$f^{-1}\in\Lin\Eucli_{\MQ}(\C)$.  Consequently, $\Lin\Eucli_{\MQ}(\C)$
is a group and $\Lin\Eucli_{\MQ}(\C)\subseteq\Lin\Eucli_{\MQ}$. Hence
$\Lin\Eucli_{\MQ}(\C)\subset\Lin\Eucli_{\MQ}$ as $\Lin\Eucli_{\MQ}$ is
not closed under composition.

$\Lin\Trivi_{\MQ}$ is a subgroup of $\Lin\Eucli_{\MQ}(\C)$,
$\Lin\Gali_{\MQ}(\C)$ and $\Lin\Poii_{\MQ}(\C)$ because
$f(\vve_{\sct})=\vve_{\sct}$ for every $f\in\Lin\Trivi_{\MQ}$.
   
To prove that $\Lin\Trivi_{\MQ}$ is a proper subgroup of these groups,
let $\varepsilon\in\C$ be such that $\varepsilon>0$.  Let
\begin{equation*}
\textstyle
p_t =
  \frac{2\varepsilon+1}{2\varepsilon^2+2\varepsilon+1},
  \quad
p_x =
  \frac{2\varepsilon^2+2\varepsilon}{2\varepsilon^2+2\varepsilon+1},
  \quad
r_t =
  \frac{2\varepsilon^2+2\varepsilon+1}{2\varepsilon+1}, 
  \quad
r_x = 
  \frac{2\varepsilon^2+2\varepsilon}{2\varepsilon+1}.
\end{equation*}
and notice that $p_t^2 + p_x^2 = r_t^2 - r_x^2 = 1$, because
$(2\varepsilon^2+2\varepsilon+1)^2-(2\varepsilon^2+2\varepsilon)^2
=(2\varepsilon+1)^2$.

By Lemma~\ref{lem:therearetransformations}, there exist
transformations $f\in\Lin\Eucli_{\MQ}$, $g\in \Lin\Gali_{\MQ}$ and
$h\in\Lin\Poii_{\MQ}$ such that
\[
  f(\vve_{\sct})=(p_t,p_x,0,0),\quad 
  g(\vve_{\sct}) = (1,\varepsilon,0,0),\quad
  h(\vve_{\sct})=(r_t,r_x,0,0).
\]

It now follows, by Proposition~\ref{prop:infball}, that
$f(\vve_{\sct}),g(\vve_{\sct}),h(\vve_{\sct})\in B_{\C}$. To see this, note that
$0<\varepsilon^2<\varepsilon<1$, whence
\begin{itemize}
\item
  $\textstyle 1-2\varepsilon<1-2\varepsilon^2<1-\frac{2\varepsilon^2}{2\varepsilon^2+2\varepsilon+1}=\frac{2\varepsilon+1}{2\varepsilon^2+2\varepsilon+1}=p_t$ and $p_t<1$;
\item
  $\textstyle 0<p_x$ and $p_x=\frac{2\varepsilon^2+2\varepsilon}{2\varepsilon^2+2\varepsilon+1}<2\varepsilon^2+2\varepsilon<4\varepsilon$,
\end{itemize} 
so that $1-2\varepsilon < p_t < 1$ and $0 < p_x < 4\varepsilon$. Similarly,
\begin{itemize}
\item $\textstyle 1<r_t=\frac{2\varepsilon^2+2\varepsilon+1}{2\varepsilon+1}<
2\varepsilon^2+2\varepsilon+1<1+4\varepsilon$;
\item $\textstyle 0<r_x=\frac{2\varepsilon^2+2\varepsilon}{2\varepsilon+1}
<2\varepsilon^2+2\varepsilon<4\varepsilon$
\end{itemize}
whence $1<r_t<1+4\varepsilon$ and $0<r_x<4\varepsilon$.

Therefore, $f\in\Lin\Eucli_{\MQ}(\C)$, $g\in\Lin\Gali_{\MQ}(\C)$ and
$h\in \Lin\Poii_{\MQ}(\C)$, but $f,g,h\not\in\Lin\Trivi_{\MQ}$ because
$f(\vve_{\sct})_s\neq (0,0,0)$, $g(\vve_{\sct})_s\neq (0,0,0)$ and
$h(\vve_{\sct})_s\neq (0,0,0)$. Thus $\Lin\Trivi_{\MQ}$ is a proper
subgroup of groups $\Lin\Eucli_{\MQ}(\C)$, $\Lin\Gali_{\MQ}(\C)$ and
$\Lin\Poii_{\MQ}(\C)$.

Finally, let $\C'$ be any cloud that does not contain
$\varepsilon$. As we saw in the proof of
Corollary~\ref{cor:infinitely-many-clouds} we can take, \eg $\C' =
\C_{\varepsilon^2}$.  To complete the present proof, it only remains
to prove that $\Eucli_{\MQ}(\C')$, $\Gali_{\MQ}(\C')$ and
$\Poii_{\MQ}(\C')$ are proper subgroups of $\Eucli_{\MQ}(\C)$,
$\Gali_{\MQ}(\C)$ and $\Poii_{\MQ}(\C)$, respectively. Moreover, to
prove this, it is enough to show that $f(\vve_{\sct}), g(\vve_{\sct}),
h(\vve_{\sct}) \not\in B_{\C'}$.

We have
\begin{equation*}
  p_x=\frac{2\varepsilon^2+2\varepsilon}{2\varepsilon^2+2\varepsilon+1}>\frac{2\varepsilon}{2\varepsilon^2+2\varepsilon+1}>\frac{2\varepsilon}{2}>\varepsilon
\end{equation*}
and similarly 
\begin{equation*}
  r_x= \frac{2\varepsilon^2+2\varepsilon}{2\varepsilon+1}> \frac{2\varepsilon}{2\varepsilon+1}> \frac{2\varepsilon}{2}>\varepsilon.
\end{equation*}
Thus $p_x,r_x\not\in\C'$ as $\varepsilon\not\in\C'$.
Hence, by Proposition~\ref{prop:infball}, $f(\vve_{\sct}), g(\vve_{\sct}), h(\vve_{\sct}) \not\in B_{\C'}$. 
\end{proof}

Changing the units used to express spatial distance is represented by the following
\semph{space scaling bijections}. Given any $\lambda>0$, we write
$S_\lambda \colon \Q^4\to\Q^4$ for the map:
\begin{equation*}
  S_\lambda \colon (t,x,y,z)\mapsto (t,\lambda x,\lambda y, \lambda z),
\end{equation*}
and note that $S_{\lambda}^{-1}=S_{\frac{1}{\lambda}}$.

\begin{prop}\label{prop:conjugation}
  For every $\cc>0$,  
\begin{equation*}
  \begin{aligned}[c]
    S_\cc \circ\Eucli_{\MQ} \circ S_\cc^{-1} &=\cEucli_{\MQ},\\
        S_\cc \circ\Trivi_{\MQ} \circ S_\cc^{-1} &=\Trivi_{\MQ},
  \end{aligned}
  \qquad\quad
  \begin{aligned}[c]
   S_\cc \circ\Poii_{\MQ} \circ S_\cc^{-1} &=\cPoii_{\MQ},\\
    S_\cc \circ\Gali_{\MQ} \circ S_\cc^{-1} &=\Gali_{\MQ}.
   \end{aligned}
\end{equation*}
\end{prop}
\begin{proof}
Since for all $T\in\Tran_{\MQ}$ and $\cc>0$, there is
$T'\in\Tran_{\MQ}$ such that $S_{\cc}\circ T=T'\circ S_{\cc}$, it is
enough to check the statements for linear transformations.  For all
$\cc>0$ and $\vvp\in\Q^4$,
  \begin{equation*}
    \|S_\cc \vvp\|^2_\cc= \cc^2\|\vvp\|^2_1\quad \text{ and }\quad   \|S_\cc \vvp\|^2_{\cc,\mu}= \cc^2\|\vvp\|^2_{1,\mu}.
  \end{equation*}
Therefore, if $L\in\Lin\Eucli_{\MQ}$, then by
Proposition~\ref{prop:Eucl},
  \begin{equation*}
    \left\|S_\cc L S_{\frac{1}{\cc}}\vvp\right\|^2_\cc=\cc^2\left\|L S_{\frac{1}{\cc}}\vvp\right\|^2_1=\cc^2\left\|S_{\frac{1}{\cc}}\vvp \right\|^2_1=\left\|S_\cc S_{\frac{1}{\cc}}\vvp \right\|^2_\cc=\|\vvp\|^2_\cc.
  \end{equation*}
Thus, by Proposition~\ref{prop:Eucl}, $S_\cc L
S_\cc^{-1}\in\Lin\cEucl_{\MQ}$.  If we compose any orthochronous
transformation from left or right with $S_\lambda$, we get an
orthochronous transformation. Therefore, we also have $S_\cc L
S_\cc^{-1}\in\Lin\cEucli_{\MQ}$. Therefore, $S_\cc \circ\Eucli_{\MQ}
\circ S_\cc^{-1} \subseteq\cEucli_{\MQ}.$ Similarly, by
Proposition~\ref{prop:Eucl}, if $L\in\Lin\cEucli_{\MQ}$, then
  \begin{equation*}
 \cc^2\left\|S_{\frac{1}{\cc}} L S_\cc\vvp\right\|^2_1=\left\|S_{\cc}S_{\frac{1}{\cc}}L S_{\cc} \vvp\right\|^2_\cc=\left\|L S_{\cc} \vvp\right\|^2_\cc=\left\|S_{\cc}\vvp \right\|^2_\cc=\cc^2\|\vvp\|^2_1.
  \end{equation*}
So $S_{\cc}^{-1} L S_\cc\in\Lin\Eucli_{\MQ}$ and thus
$S_\cc^{-1}\circ\cEucli_{\MQ}\circ S_\cc\subseteq\Eucli_{\MQ}$, which
is equivalent to $\cEucli_{\MQ}\subseteq S_\cc \circ\Eucli_{\MQ} \circ
S_\cc^{-1}$.  Consequently,
\begin{equation}
\label{eq:euclidean}
S_\cc \circ\Eucli_{\MQ} \circ S_\cc^{-1} =\cEucli_{\MQ}.
\end{equation}
An analogous proof based on Proposition~\ref{prop:Poi} using
$\|\ldots\|^2_{\cc,\mu}$ in place of $\|\ldots\|^2_\cc$ shows that
$$S_\cc \circ\Poii_{\MQ} \circ S_\cc^{-1} =\cPoii_{\MQ}.$$

Since $S_\lambda$ takes vertical lines to vertical lines for any
$\lambda$, elements of $S_\cc\circ\Trivi_{\MQ}\circ S_\cc^{-1}$ take
vertical lines to vertical ones. On the other hand,
$S_\cc\circ\Trivi_{\MQ}\circ S_\cc^{-1}\subseteq \cEucli_{\MQ}$ by
$\Trivi_{\MQ}\subseteq\Eucli_{\MQ}$ and \eqref{eq:euclidean}.
Therefore, by Lemma~\ref{lem:trivi}, $S_\cc\circ\Trivi_{\MQ}\circ
S_\cc^{-1}\subseteq\Trivi_{\MQ}$. Analogously,
$S_\cc^{-1}\circ\Trivi_{\MQ}\circ S_\cc=
S_{\frac{1}{\cc}}\circ\Trivi_{\MQ}\circ S_{\frac{1}{\cc}}^{-1}
\subseteq\Trivi_{\MQ}$. Consequently,
\[
S_\cc \circ\Trivi_{\MQ} \circ S_\cc^{-1} =\Trivi_{\MQ}.
\]

Assume $L\in\Lin\Gali_{\MQ}$. Let $\vvp=(p_t,p_x,p_y,p_z)\in\Q^4$. 
Then, we have that
\[
(S_\cc L S_\cc^{-1}\vvp)_t = p_t
\]
since $(L\vvq)_t =q_t$ and $(S_\cc \vvq)_t=q_t$ for any $\vvq\in\Q^4$.
If $p_t=0$, then
\begin{equation*}
  \sqed{S_\cc L S_\cc^{-1} (0,p_x,p_y,p_z)}= \sqed{S_\cc
  L\left(0,\frac{p_x}{\cc},\frac{p_y}{\cc},\frac{p_z}{\cc}\right)} =\\
\cc^2\sqed{L\left(0,\frac{p_x}{\cc},\frac{p_y}{\cc},\frac{p_z}{\cc}\right)}
\end{equation*}
since $L\left(0,\frac{p_x}{\cc},\frac{p_y}{\cc},\frac{p_z}{\cc}\right)_t=0$ as $L\in\Lin\Gal_{\MQ}$ and  $\sqed{S_\cc\vvq}=\cc^2\sqed{\vvq}$ for every $\vvq\in\Q^4$ with $q_t=0$. Then
\begin{equation*}
\cc^2\sqed{L\left(0,\frac{p_x}{\cc},\frac{p_y}{\cc},\frac{p_z}{\cc}\right)}=\cc^2
\sqed{\left(0,\frac{p_x}{\cc},\frac{p_y}{\cc},\frac{p_z}{\cc}\right)}=
\sqed{(0,p_x,p_y,p_z)}
\end{equation*}
since $L\in\Lin\Gal_{\MQ}$.
Consequently,
\[
p_t=0\implies \sqed{S_\cc L S_\cc^{-1}\vvp}=\sqed{\vvp}.
\]

Therefore, by Proposition~\ref{prop:Gal}, $S_\cc L S_\cc^{-1}
\in\Gali_{\MQ}$ and analogously $S_\cc^{-1}L S_\cc=S_{\frac{1}{c}}L
S_\frac{1}{c}^{-1}\in\Gali_{\MQ}$. Therefore,
\begin{equation*}
  S_\cc \circ\Gali_{\MQ} \circ S_\cc^{-1} =\Gali_{\MQ}. 
\end{equation*}
\end{proof}

For every $\cc>0$ let us introduce the following sets of affine
transformations:
\begin{align*}
 \cEucli_{\MQ}(\C)&\de S_{\cc}\circ \Eucli_{\MQ}(\C)\circ S_\cc^{-1},\\
 \cGali_{\MQ}(\C)&\de  S_{\cc} \circ \Gali_{\MQ}(\C)\circ S_\cc^{-1},\\
  \cPoii_{\MQ}(\C)&\de  S_{\cc}\circ\Poii_{\MQ}(\C)\circ S_\cc^{-1}.
\end{align*}

\begin{thm}
\label{thm:szendvics2}
Let $\C$ be a cloud. Then
there is a subcloud $\C' \subset \C$ such that
\begin{align*}
  \Trivi_{\MQ} <& \cEucli_{\MQ}(\C') < \cEucli_{\MQ}(\C) \subset
  \cEucli_{\MQ},\\ \Trivi_{\MQ} <& \cGali_{\MQ}(\C') <
  \cGali_{\MQ}(\C)< \Gali_{\MQ},\\ \Trivi_{\MQ} <& \cPoii_{\MQ}(\C')<
  \cPoii_{\MQ}(\C)< \cPoii_{\MQ}.
 \end{align*}
\end{thm}

\begin{proof}
Immediate from Theorem~\ref{thm:szendvics1} and Proposition~\ref{prop:conjugation}.  
\end{proof}

\begin{proof}[Proof of Thm.\ref{thm:non-arch}]\label{proof:non-arch:jud}
  By Theorem~\ref{thm:szendvics2}, there are strictly descending
  countably infinite chains of subgroups such that
  \begin{align*}
    \Trivi_{\MQ} &< \ldots < \G^E_i <\ldots<\G^E_1<\G^E_0 \subset
    \cEucli_{\MQ},\\
    \Trivi_{\MQ} &< \ldots < \G^G_i <\ldots<\G_1^G<\G_0^G
    <\Gali_{\MQ},\\
    \Trivi_{\MQ} &< \ldots < \G^P_i
    <\ldots<\G^P_1<\G^P_0<\cPoii_{\MQ}.
  \end{align*}
  Groups $\G^E_i$'s, $\G^G_i$'s, and $\G^P_i$'s satisfy assumptions
  \ref{awtrf1} and \ref{awtrf2} of Theorem~\ref{thm:model}, and by
  Lemma~\ref{lem:trivi}, they also satisfy assumption \ref{awtrf3}.
  Let $\M_i^E =\M(\G_i^E)$, $\M_i^G =\M(\G_i^G)$ and
  $\M_i^P =\M(\G_i^P)$. Then, by Theorem~\ref{thm:model} and
  Proposition~\ref{prop:egyszeru}, the various $\M_i^E$, $\M_i^G$ and
  $\M_i^P$ satisfy the statement of Theorem~\ref{thm:non-arch}.
\end{proof}

\section{Open questions}
${}$\\

According to Theorem~\ref{thm:groupmodel}, whenever $\MQ$ is an ordered field and $\G$ is a group for which either
 \begin{align*}
    \Trivi_{\MQ} &< \G \subset \cEucli_{\MQ},\text{ or}\\
    \Trivi_{\MQ} &< \G \le       \Gali_{\MQ},\text{ or}\\
    \Trivi_{\MQ} &< \G \le       \cPoii_{\MQ},
 \end{align*}
there exists a model \M of \BAX over $\MQ$ such that $\W_{\M}=\G$.

By Corollary~\ref{cor-1}, if $\MQ=\R$, there are no groups strictly
between $\Trivi_{\MQ}$ and $\cEucli_{\MQ}$ ($\Gali_{\MQ}$,
$\cPoii_{\MQ}$). Contrary to this, if $\MQ$ is a non-Archimedean
field, then there are groups $\G_1$, $\G_2$, and $\G_3$ such that
\begin{align*}
    \Trivi_{\MQ} &< \G_1 \subset \cEucli_{\MQ}, \\
    \Trivi_{\MQ} &< \G_2 \le       \Gali_{\MQ},\text{ and}\\
    \Trivi_{\MQ} &< \G_3 \le       \cPoii_{\MQ}.
\end{align*}
Moreover, there are infinitely many such groups by
Theorem~\ref{thm:szendvics2}. These groups were constructed from
clouds in a very specific way. It is natural to ask whether all the
groups strictly between $\Trivi_{\MQ}$ and $\cEucli_{\MQ}$
($\Gali_{\MQ}$, $\cPoii_{\MQ}$) can be constructed from clouds the
same way. More precisely, does the following hold?
	
\begin{que}\label{question}
  Suppose $\G_1$, $\G_2$, and $\G_3$ are groups satisfying
\begin{align*}
	\Trivi_{\MQ} &< \G_1 \subset \cEucli_{\MQ}, \\
	\Trivi_{\MQ} &< \G_2 < \Gali_{\MQ}, \text{ and}\\
\Trivi_{\MQ} &< \G_3 < \cPoii_{\MQ}.  
\end{align*}
Do there exist clouds $\C_1$, $\C_2$, and $\C_3$ such that 
\begin{align*}
\G_1&=\cEucli_{\MQ}(\C_1),\\
\G_2&=\cGali_{\MQ}(\C_2), \\
\G_3&=\cPoii_{\MQ}(\C_3)?
\end{align*}
\end{que}

Even if the answer to Question~\ref{question} turns out to be negative in general, it would be interesting to discover whether such clouds exist for specific classes of ordered fields, \eg the class of real closed
fields.

It is also worth noting how the consequences of Borisov's assumptions sometimes change when we replace $\R$ with other ordered fields, and sometimes remain the same.

For example, according to Borisov's Theorem, when $\MQ$ is the field $\R$ there are only two possibilities for the set of spacetime worldview transformations: either $\W=\Gali_\R$, or $\W=\cPoii_\R$ for some $\cc>0$. As explained on page~\pageref{rem1}, however, this result remains valid if we replace $\R$ with any other Archimedean field, provided every positive number has a square root.  
\begin{que}
It remains an open question whether the existence of square roots is actually necessary in this situation.
\end{que}

On the other hand, while Theorem~\ref{thm:B0} shows that Borisov's Theorem remains valid if omit \ax{Axiom\,III}, nonetheless as explained in Remark~\ref{rem2}, if we simply retain the assumption of square roots without requiring $\MQ$ to be $\R$, we can show instead that either
\begin{align*}
\Trivi_{\MQ} &< \W\subset\cEucli_{\MQ}\text{ or}\\
\Trivi_{\MQ} &< \W\leq \Gali_{\MQ}\text{ or}\\
\Trivi_{\MQ} &< \W\leq\cPoii_{\MQ}\ \text{ for some } \cc>0,
\end{align*}
and this remains true even if we replace \ax{BA\,4} with the more
general assumption that the worldlines of inertial motions according
to inertial observers are lines (but not necessarily of finite slope).
	
\begin{que}
Once again, it remains an open question if these results remain valid if we omit the assumption
  that positive numbers have square roots.
\end{que}

\bibliographystyle{apalike}
\bibliography{LogRel2019}
\end{document}

%% file: framework.tikz
\begin{tikzpicture}[scale=1,>=latex]
\pgfmathsetmacro{\r}{0.04}

\begin{scope}[shift={(3.1,3.9)}]
\draw[fill=blue!13,blue!13] (0,0) ellipse [x radius=2,y radius=1.2] ;
\draw[blue, ultra thick] (-1.1,-1) .. controls  (0,-.8) and (0,0.8) ..  (1.1,1)node[below ]{$\wl(i')$} ;
\draw[green!80!black, ultra thick] (0,-1.2) .. controls  (0,-.8) and (0,0.8) ..  (-1.1,1) node[below ]{$\wl(i)$}  ;
\draw[red!95!black, ultra thick] (1.1,-1)  node[above]{$\wl(i'')$} .. controls  (0,-.8) and (0,0.8) ..  (0,1.2);
\draw (-1.1,1.1) node[above left] {$\Ev$};
\draw[fill] (0.1,0.16) circle (\r); 
\draw[fill] (-0.1,-0.16) circle (\r); 
\end{scope}

\draw[very thick, ->] (5.4,3.8) to[out=0,in=100] node[above right]{$\Co_{m}$} (7.4,2.3);
\draw[very thick, ->] (3,2.5) to[out=-90,in=90] node[right]{$\Co_{h}$} (3.,-.3);
\draw[very thick, ->] (0.9,3.7) to[out=180,in=80] node[above left]{$\Co_{k}$} (-1.2,2.3);

\begin{scope}[shift={(-1.2,.1)}]
\draw[->] (0,0) to (0,1.7) node[left] {$k$};
\draw[->] (0,0) to (1.7,0); 
\draw[->] (0,0) to (-1.2,-0.5);
\draw[ultra thick, blue] (-0.5,-0.5) node[below]{$\wl_k(i')$}  to (1,1.5) ;
\draw[green!80!black, ultra thick] (0,-.7) to  (-.4,1.3) node[left]{$\wl_k(i)$}  ;
\draw[red!95!black, ultra thick] (0.8,-0.6)  node[below]{$\wl_k(i'')$} to (0.8,1.5);
\draw[fill] (0.8,1.23) circle (\r); 
\draw[fill] (-0.139,-0.015) circle (\r);
\node[below right] at (-.2,-1.3) {$\Q^4$}; 
\end{scope}

\draw[very thick, ->] (2.2,-1) to[out=105,in=-20] node[above right]{$f_{kh}$} (.5,.85);

\begin{scope}[shift={(3,-2.5)}]
\draw[->] (0,0) to (0,1.7) node[left] {$h$};
\draw[->] (0,0) to (1.7,0) ; 
\draw[->] (0,0) to (-1.2,-0.5);
\draw[ultra thick, blue] (0.5,-0.5) node[below]{$\wl_h(i')$}  to (0.5,1.5);
\draw[green!80!black, ultra thick] (-0.4,-.5) to  (1,1.1) node[below right]{$\wl_h(i)$}  ;
\draw[red!95!black, ultra thick]  (-.8,.2)node[left]{$\wl_h(i'')$} to (0.8,1.4);
\draw[fill] (0.5,1.17) circle (\r); 
\draw[fill] (0.5,0.53) circle (\r); 
\node[below right] at (-.2,-1.3) {$\Q^4$}; 
\end{scope}

\draw[very thick, ->] (6,0.84) to[out=200,in=85] node[above left]{$f_{hm}$} (4.8,-1.01);

\begin{scope}[shift={(7.4,.1)}]
\draw[->] (0,0) to (0,1.7) node[left] {$m$};
\draw[->] (0,0) to (1.7,0); 
\draw[->] (0,0) to (-1.2,-0.5);
\draw[ultra thick, blue]  (1,-0.5) node[below]{$\wl_m(i')$} to (-0.5,1.2);
\draw[green!80!black, ultra thick] (-.3,-.8) node[left]{$\wl_m(i)$} to  (-.3,1.4);
\draw[red!95!black, ultra thick] (0.3,-.6)  to (1,1.2) node[right]{$\wl_m(i'')$} ;
\draw[fill] (-0.3,0.98) circle (\r); 
\draw[fill] (0.54,0.02) circle (\r); 
\node[below right] at (-.2,-1.3) {$\Q^4$}; 
\end{scope}

\node at (7,6) {$i,i',i''\in\IM$};
\node at (7,5.5) {$k,h,m\in\IOb$};

\end{tikzpicture}

%% file: ax4.tikz
\newcommand{\coordsys}[1]{
\draw[->,>=latex] (0,0) to (0,3) node[left,black] {$#1$};
\draw[->,>=latex] (0,0) to (3,0); 
\draw[->,>=latex] (0,0) to (-2,-1);
}

\tikzstyle{eltolas}=[cm={0.3,-.1,0.15,.3,(1.1,1.2)}]

\pgfmathsetmacro{\x}{8}
\pgfmathsetmacro{\y}{5.6}

\begin{tikzpicture}[scale=0.7]
\begin{scope}[shift={(0,\y)},very thick]
\coordsys{\forall k}
\begin{scope}[eltolas,thick,blue]
\coordsys{k'}
\end{scope}
\node (k) at  (2,1.5){}; 
\end{scope}

\begin{scope}[shift={(\x,\y)},very thick,blue]
\coordsys{\forall k'}
\node (k') at  (-1.2,1.5){}; 
\end{scope}

\begin{scope}[shift={(0,0)},very thick]
\coordsys{\forall h}
\begin{scope}[eltolas,dashed,thick,blue],
\coordsys{h'}
\end{scope}
\node (h) at  (2,1.5){}; 
\end{scope}

\begin{scope}[shift={(\x,0)},dashed, very thick,blue]
\coordsys{\exists h'}
\node (h') at  (-1.2,1.5){}; 
\end{scope}

\draw[->,ultra thick] (k') to[out=150,in=30]node[above]{$f_{kk'}$} (k);
\draw[->,ultra thick] (h') to[out=150,in=30]node[above]{$f_{hh'}$} (h);

\node at (\x/2,3*\y/4){$f_{kk'}=f_{hh'}$};

\end{tikzpicture}

%% file: Wk.tikz
\newcommand{\coordsys}[1]{
\draw[->,>=latex] (0,0) to (0,3) node[left,black] {$#1$};
\draw[->,>=latex] (0,0) to (3,0); 
\draw[->,>=latex] (0,0) to (-2,-1);
}

\tikzstyle{eltolas}=[cm={-0.3,-.1,-0.15,.3,(1.6,0.7)}]
\tikzstyle{eltolas2}=[cm={0.3,-.2,0.25,.3,(1.3,2.3)}]
\tikzstyle{eltolas3}=[cm={0.2,.1,-0.1,.2,(1.4,-1.2)}]

\pgfmathsetmacro{\x}{8}
\pgfmathsetmacro{\y}{5}
\pgfmathsetmacro{\s}{.6}

\begin{tikzpicture}[scale=0.4]

\draw[gray,fill=gray!10] (\x/1.5,0.8) ellipse [x radius=1.5*\x,y radius=1.4*\y];
\node at (1.5*\x,1.45*\y) {$\LargeMath{\mathbb{W}_k}$};
\begin{scope}[shift={(-.3,0)},ultra thick,scale=1.2]
\coordsys{k}
\begin{scope}[eltolas,thick,red!95!black]
\coordsys{b}
\end{scope}
\begin{scope}[eltolas2,thick,green!80!black]
\coordsys{a}
\end{scope}
\begin{scope}[eltolas3,thick,blue]
\coordsys{c}
\end{scope}
\node (bk) at  (2,-1.5){}; 
\node (tk) at  (2.3,3.3){}; 
\node (rk) at (2.7,1){};
\end{scope}

\begin{scope}[shift={(\x,\y)},very thick,scale=\s,green!80!black]
\coordsys{a}
\node (a) at  (-2,.5){}; 
\end{scope}

\begin{scope}[shift={(1.3*\x,0)},very thick,scale=\s,red!95!black]
\coordsys{b}
\node (b) at  (-1.5,1.2){}; 
\end{scope}

\begin{scope}[shift={(0.8*\x,-\y)}, very thick,scale=\s,blue]
\coordsys{c}
\node (c) at  (-1.5,1){}; 
\end{scope}

\draw[->,ultra thick] (a) to[out=180,in=60]node[above left]{$f_{ka}$} (tk);
\draw[->,ultra thick] (b) to[out=150,in=30]node[above left]{$f_{kb}$} (rk);
\draw[->,ultra thick] (c) to[out=170,in=-60]node[below left]{$f_{kc}$} (bk);

\node at (1.4*\x,-0.7*\y) {\Huge \ldots};


\end{tikzpicture}

%% file: proofdiagram.tikz
\usetikzlibrary{matrix}
\begin{tikzpicture}[scale=0.5]
  \matrix (m) [matrix of math nodes,row sep=2em,column sep=4em,minimum width=2em]
  {
     {} & {} & \mathrm{(ii)} & {} \\
     \ax{Axiom\, IV} & \mathrm{(i)} & \mathrm{(iv)} & \mathrm{(iii)} \\
	 {}& {} & \mathrm{(v)} & {} \\
};
\draw[double,->] (m-2-2) to (m-1-3);
\draw[double,<->] (m-2-1) to (m-2-2);
 \draw[double,->] (m-1-3) to (m-2-3); 
\draw [double, ->] (m-2-3) to (m-3-3);
 \draw[double,<-] (m-2-2) to (m-3-3);
\draw[double,->] (m-1-3) to (m-2-4);
\draw[double,<-] (m-3-3) to (m-2-4);
\end{tikzpicture}

%% file: cloudball.pgf
\begin{tikzpicture}[scale=2.9]
\usetikzlibrary{decorations.pathreplacing}
\pgfmathsetmacro{\c}{0.17}
\pgfmathsetmacro{\e}{0.25}
\pgfmathsetmacro{\gap}{0.1}
\pgfmathsetmacro{\size}{1.8}
\pgfmathsetmacro{\l}{0.04}

\begin{scope}[shift={(0,\size+0.5)}]
\draw[very thick] (-\size,0) to (\size,0) node[below right]{$\Q$};
\draw[fill, red,opacity=0.5] (0,0) ellipse [x radius=\e,y radius=0.05];
\draw[white, ultra thick] (0,0) ellipse [x radius=\e,y radius=0.05];
\draw[fill, green,opacity=0.5] (0,0) ellipse [x radius=\c,y radius=0.05];
\draw[white, ultra thick] (0,0) ellipse [x radius=\c,y radius=0.05];
\draw[thick] (-1,\l) to (-1,-\l) node[below]{-1}; 
\draw[thick] (0,\l) to (0,-\l) node[below]{0}; 
\draw[thick] (1,\l) to (1,-\l) node[below]{1}; 
\draw [decorate,decoration={brace,amplitude=10pt}]
(-\e,0.05) -- node [above, yshift=10] 
{$\E$} (\e, 0.05) ;
\draw (0,0) node[below right, yshift=-3,xshift=2] {$\C$} ;
\end{scope}

\draw[blue!10,fill] (-\size,-\size/4) rectangle (\size,\size);
\draw[gray] (-\size,-\size/4) grid (\size,\size);

\draw[fill, red,opacity=0.5] (0,1) circle [radius=\e];
\draw[fill, green,opacity=0.5]  (0,1) circle [radius=\c]; 

\draw[very thick, ->,>=latex] (0,0) to (0,1);
\draw[very thick, ->,>=latex] (0,0) to (1,0) node[below left]{$\vve_\scx$} ;
\draw[white, very thick] (0,1)  circle [radius=\c]; 
\draw[white, very thick] (0,1)  circle [radius=\e];

\draw (0,1) node[below left]{$\vve_\sct$} ;
\draw (-\e,1+\e) node[]{$B_\E$};
\draw (\c,1+\c) node[]{$B_\C$};
\end{tikzpicture}